%% file: main.tex
\documentclass[11pt]{article}
\usepackage{authblk}
\usepackage{amssymb}
\usepackage{amsthm}
\usepackage{mathtools}
\usepackage{amsfonts}
\usepackage{amsmath}

\title{Finite-State Relative Dimension, dimensions of A. P. subsequences
  and a Finite-State van Lambalgen's theorem} 
\author[1]{Satyadev Nandakumar}
\author[1]{Subin Pulari}
\author[1]{Akhil S}
\affil[1]{
  Department of Computer Science and Engineering\\
  Indian Institute of Technology Kanpur,
  Kanpur, Uttar Pradesh, India.
}
\affil[]{\{satyadev,subinp,akhis\}@cse.iitk.ac.in}

\usepackage{amsmath,amsthm,amssymb}
\usepackage[margin=1in]{geometry}
\usepackage[colorlinks=true]{hyperref}

\newtheorem{theorem}{Theorem} 
\newtheorem{lemma}{Lemma} 
 
\newtheorem{corollary}{Corollary}
\newtheorem{definition}{Definition}
\newtheorem*{theorem*}{Theorem}

\newcommand{\Mod}{\text{ }\mathrm{mod}\text{ }}
\newcommand{\bet}{\text{}P'}
\newcommand{\dspace}{\;\;\;\;\;\;\;\;}
\include{commands}

\begin{document}

\maketitle
\begin{abstract}
Finite-state dimension (Dai, Lathrop, Lutz, and Mayordomo (2004))
quantifies the information rate in an infinite sequence as measured by
finite-state automata. In this paper, we define a relative version of
finite-state dimension. The finite-state relative dimension
$\dim_{FS}^Y(X)$ of a sequence $X$ relative to $Y$ is the finite-state
dimension of $X$ measured using the class of finite-state gamblers
with an oracle access to $Y$. We show its mathematical robustness by
equivalently characterizing this notion using the relative block
entropy rate of $X$ conditioned on $Y$.

We derive inequalities relating the dimension of a sequence to the
relative dimension of its subsequences along any arithmetic
progression (A.P.). These enable us to obtain a strengthening of
Wall's Theorem on the normality of A.P. subsequences of a normal
number, in terms of relative dimension. In contrast to the original
theorem, this stronger version has an exact converse yielding a new
characterization of normality.

We also obtain finite-state analogues of van Lambalgen's theorem on
the symmetry of relative normality.
\end{abstract}

\pagenumbering{arabic}
\section{Introduction}

Finite-state dimension, introduced by Dai, Lathrop, Lutz and Mayordomo
\cite{Dai2001} measures the information density in a sequence as
measured by finite-state betting algorithms called finite-state
$s$-gales. Every sequence drawn from a finite alphabet has a
well-defined finite-state dimension between 0 and 1. We can view
dimension 1 sequences as being ``finite-state random'', and dimension
0 sequences as being easily predictable using finite-state $s$-gales.
This notion is mathematically robust, having several equivalent
characterizations - using finite-state gamblers \cite{Dai2001}, block
entropy rates, and finite-state compressibility ratios
\cite{bourke2005entropy}. Due to this equivalence with finite-state
compressibility, a result due to Schnorr and Stimm \cite{SchSti72}
implies that a sequence is Borel normal if and only if it has
finite-state dimension 1, establishing connections to metric number
theory and probability.

In this work, we introduce a notion of the \emph{relative}
finite-state dimension of a sequence $X$ with respect to any arbitrary
sequence $Y$. Intuitively, this is the information density in $X$ as
viewed by a finite-state machine which has ``oracle'' access to bits
from $Y$. We show that this leads to a meaningful notion of relative
finite-state dimension based on finite-state $s$-gales with oracle
access. We provide an equivalent characterization of this notion using
\emph{conditional} block entropy rates, generalizing the insight in
the work of Bourke, Hitchcock and Vindochandran
\cite{bourke2005entropy}.  This demonstrates the mathematical
robustness of our definition.  Our notion is a finite-state analogue
of conditional entropy, in a similar manner as mutual dimension
introduced by Case and Lutz \cite{case2021finite} is an analogue of mutual
information.

We use our notion of relative finite-state dimension to derive a
stronger version of Wall's theorem on subsequences along arithmetic
progressions (A.P.) of normal sequences. While Wall's theorem states
that A.P. subsequences of any normal sequence $X$ must be normal, we
show that in addition, such an A.P. subsequence must also be
relatively normal with respect to all other A.P. subsequences of $X$
with the same common difference. This version has a precise converse,
unlike Wall's original theorem. This yields a new
\emph{characterization} of normality in terms of its
A.P. subsequences.

We finally show the symmetry of relative finite-state randomness - $X$
is finite-state relatively random to $Y$ if and only if the converse
holds. This can be viewed as a finite-state version of van Lambalgen's
theorem \cite{lambalgen1987random} establishing that relative
Martin-L\"of randomness is symmetric. Since van Lambalgen's theorem
analgoues fail to hold for other randomness notions \cite{Yu2007},
\cite{Bauwens2020}, \cite{DowneyHirschfeldt2010},
\cite{Chakraborty2017}, this shows that finite-state randomness is a
rare setting where symmetry of relative randomness does hold. We show,
further that this result does not generalize to arbitrary finite-state
dimensions.

As a consequence of this analogue of van Lambalgen's theorem, we
establish that finite-state independence, introduced by Becher, Carton
and Heiber \cite{Becher2018}, implies relative finite-state
normality. The converse direction remains open.

\section{Preliminaries}

We consider a finite alphabet $\Sigma$. The set of finite strings from
$\Sigma$ is denoted as $\Sigma^*$ and the set of infinite sequences,
by $\Sigma^\infty$. For any finite string $x$, we denote its length by
$|x|$. We use the notation $x[i:j]$, where $0 \le i < j \leq |x|$, to
denote the substring of $x$ from position $i$ to $j-1$, both ends
inclusive. The character at position $n$ is denoted by $x[n]$. We use
the notation $x[:n]$ to denote $x[0:n]$ which is the prefix of $x$
containing the first $n$ characters. We use similar notations for
infinite sequences. For $\ell \in \mathbb{N}$, $\Sigma^{<\ell}$
denotes the set of strings from alphabet $\Sigma$ of length less than
$\ell$.

For any $n \in \mathbb{N}$, $[n]$ denotes the set $\{0,1,\dots,n-1\}$.
Given infinite sequences $Y_0, Y_1, \dots Y_{n-1}$, we define the
\emph{product sequence} $Y_0 \times Y_1 \times Y_2 \times \dots
Y_{n-1}$ to be the infinite sequence whose $i^\text{th}$ digit is
$(Y_0[i],Y_1[i],Y_2[i],\dots Y_{n-1}[i])$. We define the
\emph{interleaved sequence} $Y_0 \oplus Y_1 \oplus Y_2 \oplus \dots
Y_{n-1}$ to be the sequence in $\Sigma^\infty$ whose $i$
\textsuperscript{th} bit is equal to $Y_{i \Mod n}[\lfloor i/n
  \rfloor]$.
  
Given any $S \in \Sigma^* \cup \Sigma^\infty$, $d \in \N$ and any $i
\in \{0,1,2,\dots d-1\}$ let $A^{S}_{d,i}$ denote the
$i$\textsuperscript{th} A. P. subsequence of $S$ with common difference
$d$. That is,
\begin{align*}
A^S_{d,i}=S[i] \;S[i+d] \;S[i+2d] \;S[i+3d] \dots.	
\end{align*}
The sequence $A^S_{d,i}$ has atmost $\lceil \lvert S \rvert /d \rceil$
digits when $S$ is a finite string, and is infinite when $S$ is
infinite. When the sequence $S$ is clear from the context, we use the
simplified notation $A_{d,i}$ instead of $A^S_{d,i}$.
  
For any natural number $n>0$, $\mathrm{Bij}(n)$ denotes the set of
permutations on $[n]$. Logarithms in this work have base 2.
\section{Finite-State Relative Dimension}
\label{sec:finitestaterelativedimesnion}

The finite-state relative dimension of an infinite sequence $X \in
\Sigma_1^\infty$ with respect to another an infinite sequence $Y \in
\Sigma_2^\infty$ represents the finite-state asymptotic information
density of $X$ relative to Y. We formulate finite-state relative
dimension using an extension of the finite-state gambler
\cite{Dai2001} which we call as finite state relative gambler. In
contrast to the classic finite-state gambler, the finite-state
relative gambler betting on an infinite sequence $X$ is given access
to the characters of another infinite sequence $Y$. The gambler may
utilize the information obtained from the characters of $Y$ while
betting on X. To preserve the finite nature of the model, the gambler
is allowed access to only a finite window of characters of Y. This
finite window shifts forward in Y, in a sliding manner as the gambler
processes X. 
\medskip


\begin{definition}
A $k$-account \emph{finite-state relative gambler }(FSRG) with window length $\ell$ is a 5-tuple $G = (Q,\delta,\vec{\beta},q_0,\vec{c})$ where:
\begin{itemize}
\item $Q$ is a nonempty, finite set of states,
\item $\Sigma_1$ and $\Sigma_2$ are finite alphabets,
\item $\delta : Q \times \Sigma_2^\ell \times \Sigma_1\rightarrow Q$ is
  the transition function,
\item $\vec\beta : Q \times \Sigma_2^\ell \times \Sigma_1 \rightarrow
  (\mathbb{Q} \cap [0,1])^k$ is the account-wise betting function, such that $\forall q
  \in Q, y \in \Sigma_2^\ell$, $i \in \{1\dots k\}$, we have $\sum_{a \in \Sigma_1}
  \beta_i(q,y,a) = 1$,
\item $q_0 \in Q$ is the initial state, and
\item $\vec{c} = (c_1,\dots c_k)$ is the initial capital per account, generally taken to be $1/k$.
\end{itemize}
\end{definition}

Let a finite state relative gambler (FSRG) G be in state $q \in Q$.
After reading the oracle word $y \in \Sigma_2^\ell$, the gambler places
the bet $ \vec\beta(q,y,x)$ on the next character $x \in \Sigma_1$ of $X$.
After processing $x$, the input pointer to $X$ as well as $Y$ move
forward by a single position. Further, $G$ enters the state
$\delta(q,y,x)$.

\medskip

Now, we define the extended transition function $\delta^* : Q \times
\Sigma_2^* \times \Sigma^*_1 \to Q $. Given $q \in Q$, $y \in
\Sigma_2^*$ and $x \in \Sigma_1^*$, $\delta^{*}(q,y,x)$ specifies the
state the gambler reaches from state $q$ on input string $x$ and
oracle string $y$. Since the oracle input requires a \emph{$\ell$-window
	lookahead}, $\delta^{*}(q,y,x)$ is undefined unless $\lvert y \rvert =
\lvert x \rvert + \ell-1$. For $q$, $y$ and $x$ satisfying this condition
\begin{align*}
	&\delta^{*}(q,y,x) = \begin{cases}
		\delta(q', \;y[|y|-\ell:|y|], \;x[|x| - 1]) &\text{if } |x| > 1,
		\\ \delta(q,y,x) & \text{if } |x| = 1.
	\end{cases}\\
	&\text{ where } q' = \delta^*(q,y[:\lvert y \rvert-1], x[:\lvert x
	\rvert-1]) 
\end{align*}

Analogous to \cite{Dai2001}, we define the concept of an $s$-gale
induced by a FSRG and the corresponding success sets.
\begin{definition}
Let $G = (Q,\delta,\vec\beta,q_0,\vec{c})$ be an $k$-account FSRG with a window length $\ell$, and with
oracle access to $Y$. The $s-gale$ induced by the $i^{th}$ account of $G$ is the function $d_{G_i,Y}^{(s)} : \Sigma_1^*
\rightarrow [0, \infty)$ defined by 
$d_{G_i,Y}^{(s)}(\lambda) = c_i$ 
and for every $w \in \Sigma_1^*$ and $x \in \Sigma_1$, by
\begin{align*}
d_{G_i,Y}^{(s)}(wx) &= |\Sigma_1|^s \times \beta_i(q,y, x) \times d_{G_i,Y}^{(s)}(w)
\end{align*}
where  $y = Y[\lvert w \rvert - 1 : \lvert w \rvert+k-1]$ and $q =
\delta^*(q_0,Y[: \lvert w \rvert+k-1],X[:|w|])$.

\medskip

The total $s-gale$ (or simply the $s-gale$) of $G$ is the function

$$d_{G,Y}^{(s)} = \sum\limits_{i=1}^{k} d_{G_i,Y}^{(s)}$$

\end{definition}	
  
We call an $s$-gale induced by a finite-state relative gambler as a
\emph{finite-state relative $s$-gale}.

Intuitively, when the automata is in state q, after reading the oracle word $y \in \Sigma_2^\ell$ and the next character $x \in \Sigma_1$, the $s$-gale
function of the $i^{th}$ account gets multiplied by a factor of $\lvert \Sigma_1 \rvert^{s}
\times \beta_i(q,y,x) $. As remarked in \cite{Dai2001}, this corresponds
to ``betting under inflation'', since the expected value after any bet
is less than that before the bet.

\begin{definition}
We say that the finite-state relative $s$-gale $d_{G,Y}^{(s)}$
\emph{succeeds} on a sequence $X \in \Sigma_1^\infty$ if
$$\limsup_{n \rightarrow \infty} d_{G,Y}^{(s)}(X[:n]) = \infty.$$ The
\emph{success set} of $d_{G,Y}^{(s)}$ is
\begin{align*}
  S^\infty[{d_{G,Y}^{(s)}}] = \{X \in  \Sigma_1^\infty \bigm\lvert
  d_{G,Y}^{(s)} \text{ succeeds on X } \}.
\end{align*}
\end{definition}

For any $X \in \Sigma_1^\infty$ and $Y \in \Sigma_2^\infty$,
$\mathcal{G}^Y(X)$ is the set of all $s \in[0, \infty)$ such that
  there is a finite-state relative $s$-gale $d_{G,Y}^{(s)}$ for which
  $X \in S^\infty[d_{G,Y}^{(s)}]$.

\begin{definition} 
The \emph{ finite-state relative dimension} of a sequence $X \in
\Sigma_1^\infty$ with respect to a sequence $Y \in \Sigma_2^\infty$ is
   \begin{align*}
     \dim_{FS}^Y(X) = \inf \mathcal{G}^Y(X).
   \end{align*}
\end{definition}

 \begin{definition}
We say that the finite-state relative $s$-gale $d_{G,Y}^{(s)}$
\emph{strongly succeeds} on a sequence $X \in \Sigma_1^\infty$ if
	$$\liminf_{n \rightarrow \infty} d_{G,Y}^{(s)}(X[:n]) = \infty.$$ The
	\emph{strong success set} of $d_{G,Y}^{(s)}$ is
	\begin{align*}
		S^\infty_{str}[{d_{G,Y}^{(s)}}] = \{X \in  \Sigma_1^\infty \bigm\lvert
		d_{G,Y}^{(s)} \text{ strongly succeeds on X } \}.
	\end{align*}
\end{definition}

For any $X \in \Sigma_1^\infty$ and $Y \in \Sigma_2^\infty$,
$\mathcal{G}_{str}^Y(X)$ is the set of all $s \in[0, \infty)$ such
  that there is a finite-state relative $s$-gale $d_{G,Y}^{(s)}$ for
  which $X \in S^\infty_{str}[d_{G,Y}^{(s)}]$.

\begin{definition} 
	The \emph{ finite-state relative strong dimension} of a sequence $X \in
	\Sigma_1^\infty$ with respect to a sequence  $Y \in
	\Sigma_2^\infty$ is
	\begin{align*}
		\Dim_{FS}^Y(X) = \inf \mathcal{G}^Y_{str}(X).
	\end{align*}
\end{definition}

Dai et al. showed the equivalence between multi account finite state gamblers and single account finite-state gamblers while defining finite-state dimension (see Theorem 4.5 in \cite{Dai2001}). It is straightforward to verify that the proof from \cite{Dai2001} relativizes and extends to the setting of a finite-state relative gamblers as well. Hence we use the model of single account FSRG in the rest of the paper, unless specified otherwise.

\section{Relative block entropy rates and  Finite-state relative dimension}
\label{sec:relativeblockentropyratesandfsrd}

 We formulate the notion of relative disjoint block entropy rates,
 extending the definitions given by Bourke, Hitchcock and
 Vinodchandran \cite{bourke2005entropy}. Analogously, we characterize finite-state relative
 dimension using relative block entropy rates, establishing that the
 notion of relative dimension is mathematically robust. The proof in
 \cite{bourke2005entropy} uses the finite state compressor
 characterization of finite-state dimension. In contrast, we use the direct characterization using $s$-gales.
\newpage

\subsection{Relative Block entropy rates}

Let $n,k,\ell \in \mathbb{N}$, where $n = {k}{\ell}$.  Given $x \in
\Sigma_1^n \text{ and } y \in \Sigma_2^n$ and $u 
\in \Sigma_1^\ell \text{ and } v \in \Sigma_2^\ell$, let
\begin{align*}
	N(v,y) = \lvert\{0\leq i < k \; | \;
	y[i\ell : (i+1)\ell] = v \}\lvert
\end{align*}
be the number of times $v$ occurs in $\ell$ length blocks of
$y$ and let
\begin{align*}
	N(u,x \;;\; v,y) = \Big\lvert\Big\{0\leq i <
	k~|~x[i\ell :
	(i+1)\ell ] = u \; \land \; y[i\ell :
	(i+1)\ell ] = v\Big\}\Big\lvert.
\end{align*}
be the number of times $u$ and $v$ occur concurrently in the disjoint
$\ell$ length blocks of $x$ and $y$ respectively.

\medskip
The \emph{block frequency} of $v$
in $y$ is defined as
\begin{align*}
	P(v , y) = \frac{N(v,y)}{k}
\end{align*}

The \emph{joint block frequency} of $u$ in $x$ with $v$
in $y$ is defined as
\begin{align*}
P(u ,x \;;\; v , y) = \frac{N(u,x\;;\;  v,y)}{k}
\end{align*}

When $N(v,y) \neq 0$, the \emph{conditional block frequency} of $u$ in
$x$ relative to $v$ in $y$ is defined as 
\begin{align*}
	P(u ,x \;|\; v , y) = \frac{N(u,x\;;\;  v,y)}{N(v,y)}
\end{align*} 

For any $v \in \Sigma_2^\ell$ such that $N(v,y) \neq 0$, it follows
that $\sum\limits_{u \in \Sigma_1^\ell} P(u ,x \;|\; v , y)= 1$.

\medskip

\begin{definition}
For two finite sequences $x \in \Sigma_1^n \text{ and } y \in
\Sigma_2^n$, the \emph{$\ell$-length block entropy rate} of $x$
relative to $y$ is defined as:
\begin{align*}
  H_\ell(x | y) = - \frac{1}{\ell \; log(|\Sigma_1|)}
  \sum_{v \in \Sigma_2^\ell} \sum_{u \in
    \Sigma_1^\ell} P(u ,x \;;\; v , y) \; \log(P(u ,x \;|\; v , y)).
\end{align*}
with the summation taken only over $u,v$ such that $N(u,x \;;\; v,y) \neq 0$. 
\end{definition}
 
\begin{definition}
For $X \in \Sigma_1^\infty \text{ and } Y \in \Sigma_2^\infty$, the
$\ell$-length \emph{block entropy rate} of $X$ relative to $Y$ is
defined as:
\begin{align*}
  H_\ell(X|Y) = 
  \liminf\limits_{k\rightarrow \infty} H_\ell(X[:k\ell]\;|\;Y[:k\ell])
\end{align*}	
The \emph{block entropy rate} of $X$ relative to $Y$ is defined
as\footnote[1]{We later show that the $\inf\limits_{\ell}$ can be
  replaced with a stronger $\lim\limits_{\ell \rightarrow \infty}$ for
  both block entropy rate and upper block entropy rate.} :
\begin{align*}
  H(X|Y) = \inf\limits_{\ell \in \mathbf{N}} H_\ell(X|Y)
\end{align*}	 
\end{definition}

\begin{definition}
For $X \in \Sigma_1^\infty \text{ and } Y \in \Sigma_2^\infty$, the
$\ell$-length \emph{upper block entropy rate} of $X$ relative to $Y$
is defined as:
\begin{align*}
  \overline{H}_\ell(X|Y) = 
  \limsup\limits_{k\rightarrow \infty} H_\ell(X[:k\ell]\;|\;Y[:k\ell])
\end{align*}	
The \emph{upper block entropy rate} of $X$
relative to $Y$ is defined as:
\begin{align*}
  \overline{H}(X|Y) = \inf\limits_{\ell \in \mathbf{N}} \overline{H}_\ell(X|Y)
\end{align*}	
\end{definition}

\subsection{Upper bounding dimension using entropy}

For two sequences $X \in \Sigma_1^\infty \text{ and } Y \in
\Sigma_2^\infty$, we prove that $\dim_{FS}^Y(X) \leq H(X|Y)$ by
constructing a finite-state relative $s$-gale for any $s>H(X|Y)$, that
succeeds on $X$. Similarly, we prove that $\Dim_{FS}^Y(X) \leq \overline{H}(X|Y)$by constructing an $s'$- gale for any $s'>
\overline{H}(X|Y)$ that strongly succeeds on $X$.

\medskip

The following lemma gives the logarithm of the capital attained by a
finite state relative gambler $G$ betting on $\ell$ length blocks of
$X$. Given a positive joint probability distribution $\bet(x,y)$, the
conditional distribution $\bet(x|y)$ can be obtained as $ \bet(x|y) =
\frac{\bet(x,y)}{\bet(y)}$ where $\bet(y) = \sum_{u \in
  \Sigma_1^\ell}\bet(u,y)$.

\begin{lemma}
	\label{lem:gamblerfromP} 
Given a block length $\ell$ and a probability distribution $\bet(x,y)
: \Sigma_1^{\ell} \ \times \Sigma_2^{\ell} \rightarrow \mathbb{Q} \cap
(0,1)$, there exists a finite state relative gambler $G$ such that for
any $X \in \Sigma_1^\infty , Y \in \Sigma_2^\infty$, and for any $s
>0$, for the $s$-gale induced by $G$,

\begin{align*}
\log (d_{G,Y}^{(s)}(X[:k\ell]) &= k \ell \log|\Sigma_1| \Big(s +
\frac{1}{\ell \log|\Sigma_1| } \sum_{y \in \Sigma_2^\ell} \sum_{x \in
  \Sigma_1^\ell} P_{k\ell}(x,y) \log( \bet(x|y)) \Big).
\end{align*}
where $P_{k\ell}(x,y) = P(x,X[:k\ell] \;;\; y, Y[:k\ell])$.
\end{lemma}

\begin{proof} 
Consider the following  finite-state relative gambler $G =
(Q,\delta,\beta,q_0)$ with window length $\ell$, where 
	
\begin{itemize}
\item $Q =  (\Sigma_2^{ \ell } \times \Sigma_1^{< \ell }) \cup \{q_\lambda\} $ 
\item For all $(y,w) \in Q$, $y' \in \Sigma_2^\ell$ and $a \in \Sigma_1$
  \begin{align*}\delta(q_\lambda,y,a) &= (y,a)\\
    \delta((y,w),y',a) &= 
    \begin{cases}
      (y,wa),& \text {if } |w| <  \ell-1\\
      q_\lambda, & \text {if } |w| =  \ell-1\\
    \end{cases}	
  \end{align*}
  
\item  For all $(y,w) \in Q$ , $y' \in \Sigma_2^\ell$ and $a \in \Sigma_1$		
  \begin{align*}
    \beta(q_\lambda,y,a) &= \bet(a|y)\\
    \beta((y,w),y',a) &= \frac{\bet(wa|y)}{\bet(w|y)}
  \end{align*}
		
  where for $|w| < \ell$, $\bet(w|y) = \sum_{w'\in\Sigma_1^{\ell - |w|}} \bet(ww'|y)$. 
\item $q_0 = q_\lambda$		
\end{itemize}
	
After processing each $\ell$ length block of $X$, the automata resets
to state $q_\lambda$, and the bets placed on the subsequent $\ell$
characters of $X$ depend only on the characters from $Y$ received at
this state.
It can be seen that if $y \in \Sigma_2^\ell$ is the oracle word
received at state $q_\lambda$, for each $x \in \Sigma_1^\ell$, a
cumulative bet corresponding to $\bet(x|y)$ is placed on the next
$\ell$ characters of input.
	
\medskip
		
Therefore, from the definition of $s$-gale, it follows that for any $k
\in \mathbb{N}$:
		
		\begin{center}
			$\begin{aligned} d_{G,Y}^{(s)}(X[:k\ell]) =
                    \prod_{y \in \Sigma_2^\ell} \prod_{x \in
                      \Sigma_1^\ell} \{ {|\Sigma_1|^{s.\ell} \times
                      \bet(x|y)} \} ^{\hat{N}(x,y)}
			\end{aligned}$
		\end{center} 	
		
where $\hat{N}(x,y) = N(x,X[:k\ell]\;;\;y,Y[:k\ell])$, the number of
times x and y occur concurrently in the disjoint $\ell$ length blocks
of $X[:k\ell]$ and $Y[:k\ell]$ respectively. Taking logarithms,
\begin{align*}
	\log (d_{G,Y}^{(s)}(X[:k\ell]) &= \sum_{y
		\in  \Sigma_2^\ell} \sum_{x \in  \Sigma_1^\ell}
	{\hat{N}(x,y)} \big(\log|\Sigma_1| \; s\ell +   \log(
	\bet(x|y)) \big) \\ 
	& = k s \ell \log|\Sigma_1| + k . \sum_{y \in
		\Sigma_2^\ell} \sum_{x \in
		\Sigma_1^\ell}  \frac{\hat{N}(x,y)}{k} \log(
	\bet(x|y))\\ 
        & = k \ell \log|\Sigma_1| \Big(s + \frac{1}{\ell \log|\Sigma_1| }
		  \sum_{y \in  \Sigma_2^\ell} \sum_{x \in
                    \Sigma_1^\ell}   P_{k\ell}(x,y) \log( \bet(x|y)) 
		  \Big). 
\end{align*}		
\end{proof}

Lemma \ref{lem:eFarProbability} ensures that all elements in a probability space are close to a strictly positive probability distribution.

\begin{lemma}
	\label{lem:eFarProbability}
	For every probability distribution $P$ on $\Omega$ and any given $\epsilon <
	\frac{1}{\lvert \Omega \rvert}$, there exists a probability distribution $Q$ on
	$\Omega$ satisfying the following:
	\begin{enumerate}
		\item For every $\omega \in \Omega$, $Q(\omega) \geq \epsilon$.
		\item $\max\limits_{\omega \in \Omega} \left\lvert P(w) - Q(w) \right\rvert \leq
		\lvert \Omega \rvert \epsilon$.
	\end{enumerate}
\end{lemma}
\begin{proof}
	We construct the required $Q$ by increasing $P(\omega)$ at any $\omega$ with
	$P(\omega)< \epsilon$ to $\epsilon$ by \emph{shifting} sufficient amount of weight
	from those $\omega$ with $P(\omega) \geq \epsilon$. These modifications are done so
	that $Q$ is a probability distribution on $\Omega$ satisfying the required
	conditions. 
	
	We define $Q$ as follows. For any $\omega$ with $P(\omega) < \epsilon$, let
	$Q(\omega)=\epsilon$. For $\omega$ satisfying $P(\omega) \geq \epsilon$ let
	\begin{align*}
	Q(\omega)= \epsilon + (P(\omega)-\epsilon)
	\left(1-\frac{\sum\limits_{\omega':P(\omega') <
			\epsilon}(\epsilon-P(\omega'))}{\sum\limits_{\omega':P(\omega') \geq
			\epsilon}(P(\omega')-\epsilon)} \right)
	\end{align*}
	The term $\sum\limits_{\omega':P(\omega') \geq \epsilon}(P(\omega')-\epsilon)$ is
	positive because $\epsilon < \frac{1}{\lvert \Omega \rvert}$. $Q(\omega)$ defined
	above is a positive quantity greater than or equal to $\epsilon$ because
	\begin{align*}
	\sum\limits_{\omega':P(\omega') < \epsilon}(\epsilon-P(\omega')) &=
	\sum\limits_{\omega':P(\omega') < \epsilon}\epsilon - 1 +
	\sum\limits_{\omega':P(\omega') \geq \epsilon}P(\omega') \\
	&= \epsilon \lvert \Omega \rvert - \sum\limits_{\omega' : P(\omega') \geq \epsilon}
	\epsilon -1 + \sum\limits_{\omega' : P(\omega') \geq \epsilon} P(\omega') \\
	&= (\epsilon \lvert \Omega \rvert -1) + \sum\limits_{\omega' : P(\omega') \geq
		\epsilon} (P(\omega')-\epsilon) \\
	& \leq \sum\limits_{\omega' : P(\omega') \geq \epsilon} (P(\omega')-\epsilon).
	\end{align*}
	
	The last inequality follows since $\epsilon < \frac{1}{\lvert \Omega \rvert}$. It is
	straightforward to verify that $Q$ is a probability distribution on $\Omega$ and it
	satisifies
	\begin{align*}
	\max\limits_{\omega : P(\omega) < \epsilon} \lvert P(\omega)-Q(\omega) \rvert \leq
	\epsilon.
	\end{align*}
	Consider any $\omega$ such that $P(\omega) \geq \epsilon$. We have
	\begin{align*}
	\lvert P(\omega)-Q(\omega) \rvert &= \frac{P(\omega)-\epsilon}{\sum\limits_{\omega'
			: P(\omega') \geq \epsilon} (P(\omega')-\epsilon)} \sum\limits_{\omega':P(\omega')
		< \epsilon}(\epsilon-P(\omega')) \\&\leq \sum\limits_{\omega':P(\omega') <
		\epsilon}(\epsilon-P(\omega')) \leq \epsilon \lvert \Omega \rvert.
	\end{align*}
	The proof of the lemma is thus complete.
\end{proof}

 In Lemma \ref{lem:deltaP}, we show that when all the bets placed are
 strictly positive, a particular Kullback-Leibler Divergence score
 between the empirical probability distribution and a close enough bet
 is small. In Lemma \ref{lem:dimlessthanH} we use this condition to
 show that when the relative entropy falls, the $s$-gale corresponding
 to a gambler that places a close enough bet to the empirical
 distribution succeeds.
 
\begin{lemma}
		\label{lem:deltaP} 
	Given $\epsilon >0$, and a probability distribution $Q(x,y)$ defined over $\Omega = \Sigma_1^\ell \times \Sigma_2^\ell$. If for all $x,y \in \Sigma_1^\ell \times \Sigma_2^\ell$  $Q(x,y) \geq \epsilon$, then for all $P(x,y)$ such that $\max\limits_{x,y} |P(x,y) - Q(x,y)| < 2\lvert \Omega \rvert \epsilon$, we have
	\begin{align*}
		\frac{1}{\ell \log|\Sigma_1|}   \sum_{y \in \Sigma_2^\ell} \sum_{x \in \Sigma_1^\ell} & P(x,y) \; [\; \log(P(x|y)) - \log{Q}(x|y)\;]\;   <  2 \lvert \Omega \rvert^4 \epsilon.
	\end{align*}
	
\end{lemma}

\begin{proof}
	It is easy to see that the term on the left expands to
\begin{align*}
	\frac{1}{\ell \log|\Sigma_1|}   \sum_{y \in \Sigma_2^\ell} \sum_{x \in \Sigma_1^\ell}  P(x,y)  \log\Big(\frac{P(x,y)}{Q(x,y)}\Big) - \frac{1}{\ell \log|\Sigma_1|}   \sum_{y \in \Sigma_2^\ell}  P(y)  \log\Big(\frac{P(y)}{Q(y)}\Big) 
\end{align*}
	 As the second term is the KL divergence $D(P(y) || Q(y)) $, which is non-negative \cite{CoverThomas1991}, we have that the above expression is less than or equal to the first term, $ D_{KL}(P(x,y) || Q(x,y))$.
	
	Define $\alpha_Q = \min\limits_{x,y} Q(x,y)$. Using the inverse of Pinsker's Inequality \cite{Gotze19}, we see that
	\begin{align*}
		D_{KL}(P(x,y) || Q(x,y)) \leq \frac{2}{\alpha_Q} \delta(P(x,y), Q(x,y))^2
	\end{align*}
where $\delta(P(x,y), Q(x,y))$ is the total variation distance \cite{DaYu17} and
it can be expressed as $\frac{1}{2}$ $\sum\limits_{x,y} |P(x,y) -
Q(x,y)|$.
	
	Using the conditions given, it follows that $\delta(P(x,y), Q(x,y)) < \lvert \Omega \rvert^2 \epsilon  $ and $\alpha_{Q} \geq \epsilon$. Therefore,
	\begin{align*}
		D_{KL}(P(x,y) || Q(x,y)) < 2 \lvert \Omega \rvert^4 \epsilon.
	\end{align*}
\end{proof}

\begin{lemma} 
	\label{lem:dimlessthanH} 
	
	For every $X \in \Sigma_1^\infty \text{ and } Y \in
	\Sigma_2^\infty$,  $\dim_{FS}^Y(X) \leq H(X|Y) $.
	
\end{lemma}
	
\begin{proof}
	
It suffices to show that for any $ s > {H}(X|Y)$, there exists an
$s$-gale $d_{G,Y}^{(s)}$ that succeeds on $X$. By definition, there
exists a block length $ \ell \in \mathbb{N}\;$ such that $
\liminf\limits_{k\rightarrow \infty} H_\ell(X[:k\ell]\;|\;Y[:k\ell]) <
s'$ for some $s' < s$.
	
	Let $S = \{P(x,y): \Sigma_1^\ell \times \Sigma_2^\ell \rightarrow [0,1] : \sum_{x,y} P(x,y) = 1\} $ be the set of joint probability distributions defined over $x,y \in \Omega$, where $\Omega = \Sigma_1^\ell \times \Sigma_2^\ell$.

	Given $0 < \epsilon < \frac{1}{\lvert \Omega \rvert}$, consider the set $T = \{{P}(x,y) \in S: \forall x,y \; P(x,y) \geq \epsilon\}$. From Lemma \ref{lem:eFarProbability} it follows that the 
	union of $2\epsilon\lvert \Omega \rvert$ open boxes around elements of $T$ cover the compact set S, that is 
	
	$$\bigcup_{Q(x,y) \in T}{\{P(x,y) : \max\limits_{x,y} |P(x,y) - Q(x,y)| < 2\epsilon\lvert \Omega \rvert}\} = S.$$
	
	Choose $\{{Q_i}(x,y) : 1 \leq i \leq \mathbf{n}\}$ to be the set of representative elements in the finite subcover of this open cover of S.
	Consider the $\mathbf{n}$-account finite  state relative gambler, with each account placing bets corresponding to a particular $Q_i(x,y)$. Proceeding with the construction in Lemma \ref{lem:gamblerfromP}, in the account corresponding to $Q(x,y)$, we have that 
	\begin{align*}
		\log ((d_{G,Y}^{(s)}(X[:k\ell])) &= 
	 k \ell \log|\Sigma_1| \Big(s + \frac{1}{\ell \log|\Sigma_1| }
	\sum_{y \in  \Sigma_2^\ell} \sum_{x \in  \Sigma_1^\ell}   P_{k\ell}(x,y) \log( Q(x|y))
	\Big)
	\end{align*}
	
	where $P_{kl}(x,y) = P(x, X[:k\ell] \;;\; y, Y[:k\ell]).$ 
	
	\medskip
	
	Since $\liminf\limits_{k\rightarrow \infty} H_\ell(X[:k\ell]\;|\;Y[:k\ell]) < s'$, it follows that for infinitely many $k
	\in \mathbb{N}$, $$-\frac{1}{\ell \log|\Sigma_1| }
	\sum_{y \in  \Sigma_2^\ell} \sum_{x \in  \Sigma_1^\ell}   P(x, X[:k\ell] \;;\; y, Y[:k\ell]) \log( P(x, X[:k\ell] \;|\; y, Y[:k\ell])) < s' .$$
Applying Lemma \ref{lem:deltaP}, we see that in the account corresponding to $Q_i(x,y)$ such that $\max\limits_{x,y} |P(x,y) - Q(x,y)| < 2\lvert \Omega \rvert \epsilon$, we have
	\begin{align*}\log ((d_{G,Y}^{(s)}(X[:k\ell])) > 
		k  \ell \log|\Sigma_1| \; \big(s  - s ' - 2 \lvert \Omega \rvert^4 \epsilon\big). 
	\end{align*}	
	Choose an $\epsilon < \frac{1}{\lvert \Omega \rvert}$ such that $s > s' + 2 \lvert \Omega \rvert^4 \epsilon$. Since $k$ is unbounded, the $\mathbf{n}$-account $s-gale$ succeeds on X relative to Y. 
\end{proof}

\begin{lemma} 
	\label{lem:DimlessthanUpperH} 
	
	For every $X \in \Sigma_1^\infty \text{ and } Y \in
	\Sigma_2^\infty$,  $\Dim_{FS}^Y(X) \leq \overline{H}(X|Y) $.
	
\end{lemma}
\begin{proof}
	
	For any $ s> \overline{H}(X|Y)$, there exists a block length $ \ell \in \mathbb{N}\;$ such that $ \limsup\limits_{k\rightarrow \infty}  H_\ell(X[:k\ell]\;|\;Y[:k\ell]) < s'$ for some $s' < s$ . 
	
	Proceeding in the same lines as Lemma \ref{lem:dimlessthanH}, we get an $\mathbf{n}$-account finite state gambler \cite{Dai2001}, such that for sufficiently large $k\in\mathbb{N}$, in one among the accounts,
	
		\begin{align*}\log ((d_{G,Y}^{(s)}(X[:k\ell])) >
		k  \ell \log|\Sigma_1| \big(s  - s ' - 2 \lvert \Omega \rvert^4 \epsilon\big). 
	\end{align*}	
Choose an  $\epsilon < \frac{1}{\lvert \Omega \rvert}$ such that $s > s' + 2 \lvert \Omega \rvert^4
\epsilon$. Hence the $\mathbf{n}$ account finite state gambler
strongly succeeds on $X$ relative to $Y$. 
\end{proof} 

\subsection{Upper bounding entropy using dimension}

Let an $s$-gale $d_{G,Y}^{(s)}$ corresponding to a finite-state
gambler $G$ succeed on $X \in \Sigma_1^\ell$. For a block length $L
\in \mathbb{N}$, we construct the \emph{stretched finite-state
relative gambler} $G_{L}$, to keep track of the state encountered
after every run of input of length $L$ on $G$.

\begin{definition}
		\label{lem:GLfromG} 
	Let $G = (Q,\delta,\beta,q)$ be a finite state relative
        gambler with window length $l$. For any $L \in \mathbb{N}$
        such that $L > l$, we define the \emph{$L$-length stretched
        FSRG} of $G$ to be $G_L = (Q',\delta',\beta',q_0')$, where:
	
	\begin{itemize}
		\item $Q' = Q \times [L] .$
		\item For all $(q,n) \in Q'$ , $y \in \Sigma_2^l$ and $x \in \Sigma_1$,
		
		\[
		\delta'((q,n),y,x) = 
		\begin{cases}
			(\delta(q,y,x),n+1)& \text {if } n < L-1\\
			(\delta(q,y,x),0) & \text {if } n = L-1\\
		\end{cases}
		\]
		
		\item For all $(q,n) \in Q'$ , $y \in \Sigma_2^l$ and $x \in \Sigma_1$,
		
		\[
		\beta'((q,n),y,x) = 
		\begin{cases}
			\beta(q,y,x),& \text {if } n \leq L - l\\
			1/|\Sigma_1|, & \text {if } n > L- l \\
		\end{cases}
		\]
		
		\item $q_0' = (q_0,0)$.
		
	\end{itemize}

\end{definition}

\begin{lemma}
	\label{lem:ss'}
If the $s$-gale $d_{G,Y}^{(s)}$ succeeds on $X$, then for $s' =
s+\frac{l-1}{L} $, the s'-gale corresponding to the stretched FSRG,
$d_{G_L,Y}^{(s')}$ succeeds on $X$.
\end{lemma}
\begin{proof}
	 $d_{G_L,Y}^{(s')}$ succeeds on $X$, because after processing $k$ blocks of length $L$,
	\begin{align*}
		d_{G_L,Y}^{(s')}(X[:kL]) ) &= |\Sigma_1|^{kL(s+\frac{l-1}{L})} \prod_{i=1}^{k} \left\{ \prod_{j = 0}^{L-l}\beta_{(i,j)} \prod_{j = L-l+1}^{L-1} (1/|\Sigma_1|) \right\} \\
		&= \frac{1}{|\Sigma_1|^{k.(l-1)}} |\Sigma_1|^{k(l-1)} |\Sigma_1|^s \prod_{i=1}^{k} \left\{ \prod_{j = 0}^{L-l}\beta_{(i,j)} \prod_{j = L-l+1}^{L-1} 1 \right\} \\
		&\geq d_{G_L,Y}^{(s)}(X[:kL]).
	\end{align*}
\end{proof}

We show that  if $s$-gale $d_{G,Y}^{(s)}$ succeeds on $X$, then for any sufficiently large block length $L$, $H_L(X|Y) \leq s + c/L$, where c is a constant which does not depend on $L$. 

\medskip

\begin{lemma}
	\label{lem:Hlessthandim}
	For every $X \in \Sigma_1^\infty \text{ and } Y \in
        \Sigma_2^\infty$, we have $H(X|Y)  \leq \dim_{FS}^Y(X)$.
\end{lemma}
\begin{proof}
	Consider any $s$ such that $s > \dim_{FS}^Y(X)$. It suffices to show that $H(X|Y) \leq s$.
	By definition, there exists a finite-state relative $s$-gale $d_{G,Y}^{(s)}$  that succeeds on
         $X$.  Let $l$ be the oracle window length of the corresponding finite-state relative gambler $G$. For any $L \in \mathbb{N}$ such that $L > l$, consider the \emph{$L$-length stretched FSRG} $G_L$ given in Definition \ref{lem:GLfromG}. From Lemma \ref{lem:ss'}, it follows that taking $s' =
         s+\frac{l-1}{L} $, the s'-gale corresponding to $G_L$,  $d_{G_L,Y}^{(s')}$ succeeds on $X$.
	\medskip

		After processing each $L$ length block of characters from input $X$, $G_L$ reaches the state $(q,0)$ for some $q \in Q$. Let $x = X[iL:(i+1)L] $, the $i^{th}$ disjoint block of length $L$ in $X$ for some $i\in \mathbb{N}$. The bets placed by $G_L$ on the last $L-l$ characters of $x$ is constant. Therefore, the cumulative bet $\beta(q,y,x)$ placed by $G_L$ on $x$ is dependent only on the state $(q,0)$ encountered while processing the block and the corresponding $L$ characters $y = Y[iL:(i+1)L] $ of the oracle input Y. For any $q\in Q, x\in \Sigma_1^L \text{ and } y\in\Sigma_2^L$, let $N(q,y,x)$ be the number of times the an $L$ length block in $X$ is $x$,  in $Y$ is $y$ and the stating state of the block was $(q,0)$. Therefore,
	
	\begin{align*}
		d_{G_L,Y}^{(s')}(X[:kL]) =  \prod_{q \in Q} \prod_{y \in  \Sigma_2^L} \prod_{x \in  \Sigma_1^L} \{ {|\Sigma_1|^{s'.L} \times \beta(q,y,x)} \} ^{N(q,y,x)}.
	\end{align*}

	 Take $N(q,y) = {\sum\limits_{x \in  \Sigma_1^L}N(q,y,x)}$ and $N(q) = {\sum\limits_{y \in  \Sigma_2^L}N(q,y)}$. Hence,
	\begin{align*}
		\log (d_{G_L,Y}^{(s')}(X[:kL])) &= s'kL \log|\Sigma_1| + \sum_{q\in Q} \sum_{y \in  \Sigma_2^L} \sum_{x \in  \Sigma_1^L} {N(q,y,x)} \; \log( \beta(q,y,x)) \\
		&= s'kL \log|\Sigma_1|  \\ 
		&\dspace\;\; +\; k\; 
		{\sum_{q \in Q}} \frac{N(q)}{k}
		{\sum_{y \in  \Sigma_2^L}} \frac{N(q,y)}{N(q)}
		{\sum_{x \in  \Sigma_1^L}}\frac{N(q,y,x)}
		{N(q,y)} \log( \beta(q,y,x)).
	\end{align*}
	
	Taking $P(q) = \frac{N(q)}{k}$, $P(y|q) = \frac{N(q,y)}{N(q)}$ and $P(x|y,q) = \frac{N(q,y,x)} {N(q,y)}$ , stipulating that these take values $0$ if the denominators are zero,
	
	\begin{align*}
		\log (d_{G_L,Y}^{(s')}(X[:kL])) 
		&= k \Big( s'L \log|\Sigma_1| \\ 
		&\dspace \dspace + \sum_{q \in Q} P(q)  \sum_{y \in  \Sigma_2^L} P(y|q)
		\sum_{x \in  \Sigma_1^L} P(x|y,q) \; \log(
		\beta(q,y,x))\Big) .
	\end{align*}

	$d_{G_L,Y}^{(s')}$ succeeds on $X$ only if for infinitely many $k \in \mathbb{N}$,
	
	\begin{align*}
		-  \sum_{q \in Q} P(q)  \sum_{y \in  \Sigma_2^L} P(y|q) 
		\sum_{x \in  \Sigma_1^L} P(x|y,q) \; \log(
		\beta(q,y,x))  < s'L\log|\Sigma_1|.
	\end{align*}
	
	For all $q\in Q$ and ${y \in \Sigma_2^L}$,  $\sum\limits_{x \in \Sigma_1^L} \beta(q,y,x) = 1$ and $\sum\limits_{x \in \Sigma_1^L} P(x,y|q) = 1$ .
	
	Therefore, the KL divergence of these distributions \cite{CoverThomas1991} $	D_{KL}(P(x|y,q) || \beta(q,y,x))$ over values in $x$ is well defined and is non-negative. So it follows that
	
	\begin{align*}
	- \sum_{q \in Q} P(q) \sum_{y \in  \Sigma_2^L} P(y|q) 
	\sum_{x \in  \Sigma_1^L} P(x|y,q) \; \log( P(x|y,q)) < s'L\log|\Sigma_1|.
	\end{align*}

	Using result from information theory that $H(X|Y) \leq H(X|Q,Y) + H(Q|Y)$ \cite{CoverThomas1991}, it follows that,
	
	\begin{align*}
		 - \sum_{y \in  \Sigma_2^L} P(y) \sum_{x \in  \Sigma_1^L} P(x|y) \; \log( P(x|y)) &< s'L\log|\Sigma_1| + H(Q|Y). 
	\end{align*}
	
	$H(Q|Y)$ is at most $\log(n)$, where $n = |Q|$, the number of states in the original gambler $G$.

	Dividing by $L \log|\Sigma_1|$ and substituting the value of $s'$, we have that for infinitely many $k \in \mathbb{N}$, 
	\begin{align*}
	H_L(X[:kL] \;|\; Y[:kL])\ < \;   s + \frac{l-1}{L} + \frac{\log(n)}{L\log|\Sigma_1|}.
	\end{align*}
	
	Taking $\liminf$ over k,
	\begin{equation}\label{eq:HLandL}
		H_L(X|Y) \leq \;   s + \frac{l-1}{L} + \frac{\log(n)}{L\log|\Sigma_1|} .
	\end{equation}
	
	Taking infimum over all $L$,
	   \begin{align*}
	   	H(X|Y) \leq s.
	   \end{align*}
	
\end{proof}

\begin{lemma} 
	\label{lem:UpperHlessthanDim} 
	
	For every $X \in \Sigma_1^\infty \text{ and } Y \in
	\Sigma_2^\infty$,  $ \overline{H}(X|Y) \leq \Dim_{FS}^Y(X) $.	
\end{lemma}

\begin{proof}
	The proof follows along the lines of the proof of Lemma \ref{lem:Hlessthandim}. Using the condition that an $s$-gale strongly succeeding on $X$,  we have that for sufficiently large $k \in \mathbb{N}$, 
	
	\begin{equation}\label{eq:HLandL2}
		H_L(X[:kL] \;|\; Y[:kL])\ < \;   s + \frac{l-1}{L} + \frac{\log(n)}{L\log|\Sigma_1|}.
	\end{equation}
	
	From which it follows that,
	\begin{align*}
		\overline{H}(X|Y) \leq s.
	\end{align*}
\end{proof}

We originally defined {block entropy rates} by taking infimum of $H_\ell(X|Y)$ over all block lengths $\ell$. But as a consequence of the proof, we show that the infimum is infact a limit as $\ell \rightarrow \infty$.

\begin{theorem}
\label{thm:fsrelativedimensionentropyrate} For any two sequences $X \in \Sigma_1^\infty$ and $Y \in \Sigma_2^\infty$,
	\[\dim_{FS}^Y(X) =	\lim_{\ell \rightarrow \infty} H_\ell(X|Y).\]
\end{theorem}

\begin{proof}
	Lemma \ref{lem:dimlessthanH} established that $\inf\limits_\ell H_\ell(X | Y ) \geq \dim_{FS}^Y(X)$, therefore
	
	\begin{equation}
	\liminf_{\ell \rightarrow \infty} H_\ell(X|Y) \geq \dim_{FS}^Y(X).	\label{eq:3}
	\end{equation}
	
	From \eqref{eq:HLandL} in Lemma \ref{lem:Hlessthandim}, it follows that
	\begin{equation}
	\limsup_{\ell \rightarrow \infty} H_\ell(X|Y) \leq \dim_{FS}^Y(X). \label{eq:2}
	\end{equation}

	The theorem follows immediately from \eqref{eq:3} and \eqref{eq:2}.
\end{proof}

\medskip

\begin{theorem}
	\label{thm:fsstrongrelativedimensionentropyrate} For any two sequences $X \in \Sigma_1^\infty$ and $Y \in \Sigma_2^\infty$,	
	\[\Dim_{FS}^Y(X) =  \lim_{\ell \rightarrow \infty} \overline{H_\ell}(X|Y).\]
	
\end{theorem}

\begin{proof}
	It follows from Lemma \ref{lem:DimlessthanUpperH} that
	\begin{equation}
	\liminf_{\ell \rightarrow \infty} \overline{H}_\ell(X|Y) \geq \Dim_{FS}^Y(X).	\label{eq:5}
	\end{equation}
	
		From \eqref{eq:HLandL2} in Lemma \ref{lem:UpperHlessthanDim}, it follows that
	\begin{equation}
	\limsup_{\ell \rightarrow \infty} \overline{H}_\ell(X|Y) \leq \Dim_{FS}^Y(X). \label{eq:6}
	\end{equation}
	
	The theorem follows immediately from \eqref{eq:5} and \eqref{eq:6}.
\end{proof}

\medskip
\subsection{Multiple Oracles}
\medskip

Consider the case when the automaton has access to multiple oracle
sequences, say $Y_1\cdots Y_n$, where each $Y_i \in \Sigma_2^\infty$.
The finite-state relative dimension and relative block entropy rates
in this scenario is defined by considering a single oracle sequence $Y
\in (\Sigma_2^n)^\infty$, such that $Y = Y_1 \times Y_2 \dots \times
Y_n$.

\begin{definition}
The finite-state relative dimension of a sequence $X$ with respect to
a sequences ${Y_1 \cdots Y_n}$ is defined as $\dim_{FS}^{Y_1 \cdots
  Y_n}(X) = \dim_{FS}^{Y}(X)$.
\end{definition}


\begin{definition}
The \emph{relative block entropy rate} of $X$ with respect to $Y_1$,
$\dots$, $Y_n$ is defined as $$H(X|Y_1 \cdots Y_n) = H(X|Y).$$
\end{definition}
 Therefore from Theorem
 \ref{thm:fsrelativedimensionentropyrate} it follows that,
\begin{theorem}
\label{thm:multifsrelativedimensionentropyrate}
For sequences $X \in \Sigma_1^\infty$ and $Y_1 \dots Y_n \in \Sigma_2^\infty$
\begin{align*}
	\dim_{FS}^{Y_1 \cdots Y_n}(X) = H(X|Y_1 \cdots Y_n)\\
	\Dim_{FS}^{Y_1 \cdots Y_n}(X) = \overline{H}(X|Y_1 \cdots Y_n)
\end{align*}

\end{theorem}

\subsection{Properties of of finite-state relative dimension}

\medskip

The following are some basic properties of finite-state relative dimension.
\begin{lemma}
\label{lem:basicproperties}
For any $X \in \Sigma_1^\infty$ , $Y \in \Sigma_2^\infty$ and $0 \in \Sigma_2$,
\begin{enumerate}
\item $\dim_{FS}^Y(X) \leq \dim_{FS}(X)$. \vspace{0.1cm}

\item $\dim_{FS}^{X}(X) = 0$.  \vspace{0.1cm}

\item $\dim_{FS}^{0^{\infty}}(X)=\dim_{FS}(X)$.

\item $\dim_{FS}^Y(X) \leq \Dim_{FS}^Y(X)$.
\end{enumerate}	
\end{lemma}
\begin{proof}
  
We first prove statement 1. For any $s > \dim_{FS}(X)$, there exists an
$s$-gale corresponding to a finite-state state gambler $G$ that
succeeds on $X$ by definition. For any $Y$, we can construct a
finite-state relative gambler $G'$ from $G$ as follows. $G'$ places
the same bets on characters of $X$ as $G$, irrespective of the bits of
oracle $Y$ obtained. Hence, for any $X \in \Sigma_1^*$ and $Y \in
\Sigma_2^\infty$, $d^{(s)}_{G',Y}(X) = d_{G}^{(s)}(X)$. Therefore,
$\dim_{FS}^Y(X) \leq s$, and the statement follows.
	
For the proof of Statement 2, consider the finite-state relative
gambler $G$ with a single state $q$ and has oracle window length one. The automata stays in state $q$ irrespective of the
input bit and oracle bits, ie $\forall x,y \in \Sigma_1 \;
\delta(q,y,x) = q $. The gambler bets all its capital on the character
received from the oracle. That is $\beta(q,y,x) = 1$ when $x = y$,
otherwise $\beta(q,y,x) = 0$. For any $X \in \Sigma_1^\infty$ and
$s>0$, the finite-state $s$-gale corresponding to $G$ with oracle
access to $X$ succeeds on $X$. Hence, it follows that
$\dim_{FS}^{X}(X) = 0$. This proves statement 2.

For the proof of Statement 3, consider an $s$-gale corresponding to a
relative gambler $G = (Q,\delta,\beta,q_0,c_0)$ with a window length
of $\ell$ that wins on a sequence $X \in \Sigma_1^\infty$, with an oracle
access to $0^\infty$. Design a gambler
$G=(Q',\delta',\beta',q_0',c_0')$ that has the same set of states as
$G$, $Q' = Q$, and $q_0' = q_0$, $c_0' = c_0$. $G'$ makes the same
transitions and bets as $G$, when the oracle bit received are $0^k$.
That is $\delta'(q,x) = \delta(q,0^\ell,x)$ and $\beta'(q,x) =
\beta(q,0^\ell,x)$. It can be seen that for any $X \in \Sigma_1^\infty$ and
$s>0$, $d_{G,0^\infty}^{(s)}(X) = d_{G'}^{(s)}(X)$ . Therefore
$\dim_{FS}^{0^{\infty}}(X) \leq \dim_{FS}(X)$. This along with
statement 1 proves statement 3.

Statement 4 follows from the fact that $\liminf a_n \leq \limsup a_n$
for any sequence $\langle a_n \rangle$. 
\end{proof}

\section{Finite-state relative dimensions of A. P. subsequences}
\label{sec:fsrelativedimensionandapsubsequences}
In this section, we study the finite-state dimensions of subsequences
chosen along arithmetic progressions (A. P.) using the information
theoretic characterization of finite-state relative dimension
developed in Theorem \ref{thm:fsrelativedimensionentropyrate}. We
derive a collection of inequalities which demonstrate the
relationships between the dimension of a sequence and the dimensions
of its A. P. subsequences. The main results in the following sections are
derived as the consequences of these inequalities. At the end of this
section, we show that these bounds are necessarily inequalities, by
providing examples. These examples also show that the bounds are
tight. For $i \leq d-1$, let $A_{d,i}(X)$ represent the A. P. subsequence
$X[i]X[i+d]X[i+2d]X[i+3d]\dots $ of the sequence $X$. When the
sequence $X$ is clear from the context we use the simplified notation
$A_{d,i}$ to refer to this subsequence.

Using the chain rule of Shannon entropy \cite{CoverThomas1991}, we obtain our first inequality which gives a lower bound for the finite-state
dimension of $X$ in terms of the finite-state relative dimensions of
the its A. P. subsequences.
\begin{lemma}
\label{lem:mainlowerboundinequality}
For any $X \in \Sigma^\infty$, $d \in \N$ and $\sigma \in
\mathrm{Bij}(d)$, 
\begin{align*}
\dim_{FS}(X) \geq \frac{1}{d} \left(\dim_{FS}(A_{d,\sigma(0)})
+\sum\limits_{i=1}^{d-1}
\dim_{FS}^{A_{d,\sigma(0)},A_{d,\sigma(1)},\dots
  A_{d,\sigma(i-1)}}(A_{d,\sigma(i)}) \right)
\end{align*}
\end{lemma}

A result we require in the proof of Lemma
\ref{lem:mainlowerboundinequality} is the following inequality which
gives a lower bound on the limit inferior of the sum of finitely many
real sequences.
\begin{lemma}
\label{lem:liminflowerbound}
Let $\langle a_{0,n} \rangle_{n =0}^{\infty},\langle a_{1,n}
\rangle_{n =0}^{\infty},\dots \langle a_{k-1,n} \rangle_{n
  =0}^{\infty}$ be any finite collection of sequences of real
numbers. Then, 
\begin{align*}
\liminf\limits_{n \to \infty} \left( \sum\limits_{i=0}^{k-1} a_{i,n}
\right)\geq \sum\limits_{i=0}^{k-1} \liminf\limits_{n \to \infty}
a_{i,n}. 
\end{align*}
\end{lemma}

Lemma \ref{lem:liminflowerbound} can be proved using routine real analytic arguments.

In order to prove Lemma \ref{lem:mainlowerboundinequality}, we require
the following relationship between the finite-state dimensions of the
product sequence and the interleaved sequence of any given finite set
of sequences.
 
\begin{lemma}
	\label{lem:fsdofproductandfsdofinterleaving}
	Let $Y_0,Y_1,Y_2 \dots Y_{n-1}$ be sequences in $\Sigma^\infty$. Then,
	for any $\sigma \in \mathrm{Bij}(n)$,
	\begin{align*}
		\dim_{FS}(Y_0 \times Y_1 \times Y_2 \times \dots Y_{n-1}) &= \dim_{FS}(Y_{\sigma(0)} \times Y_{\sigma(1)} \times Y_{\sigma(2)} \times \dots Y_{\sigma(n-1)})\\
		&=\dim_{FS}(Y_{\sigma(0)} \oplus Y_{\sigma(1)} \oplus Y_{\sigma(2)} \oplus \dots Y_{\sigma(n-1)})\\
		&=\dim_{FS}(Y_{0} \oplus Y_{1} \oplus Y_{2} \oplus \dots Y_{n-1}).
	\end{align*}	
\end{lemma}

Now, we use Theorem \ref{thm:fsrelativedimensionentropyrate}, Lemma \ref{lem:fsdofproductandfsdofinterleaving} and Lemma \ref{lem:liminflowerbound} to prove Lemma \ref{lem:mainlowerboundinequality}.
\begin{proof}[Proof of Lemma \ref{lem:mainlowerboundinequality}]
	Let us first observe that $X=A_{d,0} \oplus A_{d,1} \oplus \dots A_{d,d-1}$. Therefore from Lemma \ref{lem:fsdofproductandfsdofinterleaving}, it follows that,
	\begin{align*}
	\dim_{FS}(X) &= \dim_{FS}(A_{d,0} \oplus A_{d,1} \oplus \dots A_{d,d-1})\\
	&= \dim_{FS}(A_{d,\sigma(0)} \oplus A_{d,\sigma(1)} \oplus \dots A_{d,\sigma(d-1)})\\
	&= 	\dim_{FS}(A_{d,\sigma(0)} \times A_{d,\sigma(1)} \times \dots A_{d,\sigma(d-1) }). 
	\end{align*}
	From the entropy rate characterization of finite-state dimension \cite{bourke2005entropy} it follows that,
	\begin{align}
	\label{eqn:mainlowerboundeqn1}
	\dim_{FS}(A_{d,\sigma(0)} \times  \dots \times A_{d,\sigma(d-1)}) &= \lim\limits_{l \to \infty} \liminf\limits_{k \to \infty} H_l(A_{d,\sigma(0)} [:kl],  \dots A_{d,\sigma(d-1)} [:kl]).
	\end{align}
	Since, $A_{d,\sigma(0)} \times A_{d,\sigma(1)} \times \dots A_{d,\sigma(d-1)}$ is a sequence in $(\Sigma^d)^\infty$, the normalization factor in the $H_l$ on the right hand side is $dl$ instead of $l$. Using the chain rule of Shannon entropy (see \cite{CoverThomas1991}) it follows that,
	\begin{align}
	\label{eqn:mainlowerboundeqn2}
	&dl H_l(A_{d,\sigma(0)} \times A_{d,\sigma(1)}\times \dots A_{d,\sigma(d-1)}[:kl])	= l H_l(A_{d,\sigma(0)}[:kl]) \nonumber \\
	&+ \sum\limits_{i=1}^{d-1} l H_l(A_{d,\sigma(i)}[:kl] \mid A_{d,\sigma(0)}[:kl], A_{d,\sigma(1)}[:kl] \dots A_{d,\sigma(i-1)}[:kl] ).
	\end{align}
	Dividing both sides by $dl$ and substituting the above in \ref{eqn:mainlowerboundeqn1}, we get,
	\begin{align*}
	  &\dim_{FS}(X) = \dim_{FS}(A_{d,\sigma(0)} \times
          A_{d,\sigma(1)} \times \dots A_{d,\sigma(d-1)})\\ 
	  &= \frac{1}{d}
          \lim\limits_{l \to \infty} \liminf\limits_{k \to \infty}
          \Bigg(H_l(A_{d,\sigma(0)}[:kl]) \\
          &\qquad\qquad\qquad +
          \sum\limits_{i=1}^{d-1}  H_l(A_{d,\sigma(i)}[:kl]\mid A_{d,\sigma(0)}[:kl], \dots
          A_{d,\sigma(i-1)}[:kl] )\Bigg) \\ 
	  &\geq \frac{1}{d}
          \lim\limits_{l \to \infty}  \Bigg(\liminf\limits_{k \to
            \infty}H_l(A_{d,\sigma(0)}[:kl]) \\
          &\qquad\qquad\qquad +
          \sum\limits_{i=1}^{d-1}  \liminf\limits_{k \to
            \infty}H_l(A_{d,\sigma(i)}[:kl] \mid
          A_{d,\sigma(0)}[:kl],  \dots
          A_{d,\sigma(i-1)}[:kl] )\Bigg) \\ 
	  &= \frac{1}{d}
          \Bigg(\lim\limits_{l \to \infty}\liminf\limits_{k \to
            \infty} H_l(A_{d,\sigma(0)}[:kl])  \\
          &\qquad\qquad\qquad +
          \sum\limits_{i=1}^{d-1}  \lim\limits_{l \to
            \infty}\liminf\limits_{k \to
            \infty}H_l(A_{d,\sigma(i)}[:kl] \mid
          A_{d,\sigma(0)}[:kl], \dots
          A_{d,\sigma(i-1)}[:kl] )\Bigg). 
	\end{align*}
	The inequality in the third line above is  obtained from Lemma \ref{lem:liminflowerbound}. Since, 
	\begin{align*}
	\dim_{FS}(A_{d,\sigma(0)}) = \lim\limits_{l \to \infty}\liminf\limits_{k \to \infty}H_l(A_{d,\sigma(0)}[:kl])	
	\end{align*}
	and,
	\begin{multline*}
	\dim_{FS}^{A_{d,\sigma(0)},A_{d,\sigma(1)},\dots
          A_{d,\sigma(i-1)}}(A_{d,\sigma(i)}) = \\
        \lim\limits_{l \to \infty}\liminf\limits_{k \to \infty}
        H_l(A_{d,\sigma(i)}[:kl] \mid
        A_{d,\sigma(0)}[:kl], \dots
        A_{d,\sigma(i-1)}[:kl] ), 
	\end{multline*}
	it follows that
	\begin{align*}
	\dim_{FS}(X) \geq \frac{1}{d} \left(\dim_{FS}(A_{d,\sigma(0)}) +\sum\limits_{i=1}^{d-1} \dim_{FS}^{A_{d,\sigma(0)},A_{d,\sigma(1)},\dots A_{d,\sigma(i-1)}}(A_{d,\sigma(i)}) \right).
	\end{align*}
\end{proof}

By selecting $\sigma$ to be
the identity mapping, we obtain the following corollary of Lemma \ref{lem:mainlowerboundinequality}.
\begin{corollary}
\label{cor:maininequalitycorollary1}
Consider any $X \in \Sigma^\infty$ and $d \in \N$. If for some $r \in
[0,1]$, $\dim_{FS}(A_{d,0}) \geq r$ and for every i, 
\begin{align*}
	\dim_{FS}^{A_{d,0},A_{d,1},\dots A_{d,i-1}}(A_{d,i}) \geq r.
\end{align*}
Then, $\dim_{FS}(X) \geq r$. 
\end{corollary}

It follows from Theorem \ref{thm:strongerwallstheoremwithconverse} in Section
\ref{sec:strongerwallstheoremwithconverse} 
that when $X$ is a normal sequence, 
the bound in Lemma \ref{lem:mainlowerboundinequality} is tight. We show that inequality in Lemma \ref{lem:mainlowerboundinequality} cannot be replaced with an
equality in general.  We state the following
lemma for the case when $d=2$ and $\sigma$ is the identity mapping
from $\{0,1\}$ to itself, for the sake of simplicity. Using ideas from \cite{lutz2021weyl}, we obtain the following lemma which generalizes to arbitrary $d$
and $\sigma$ in a straightforward manner.

\begin{lemma}
\label{lem:upperboundistight}
	There exists an $X \in \Sigma^\infty$ such that,

	\[\dim_{FS}(X)>\frac{1}{2}\left(\dim_{FS}(A_{2,0})+\dim_{FS}^{A_{2,0}}(A_{2,1}) \right) \] 	
\end{lemma}
\begin{proof}
	
	Fix a normal number $Y = Y[0]Y[1]Y[2]\dots$. Let $S_0$ be the infinite string $0Y[0]0Y[1]0Y[2]\dots$ and let $S_1$ be the infinite string $Y[0]0Y[1]0Y[2]0\dots$. The required string $X$ is constructed in a stagewise manner, where in alternate stages we append prefixes of $S_0$ and $S_1$ to the end of the prefix $X$ constructed until the previous stage. In even numbered stages where a long enough prefix of $S_0$ is attached to the end, the block entropy rate of $A_{2,0}^X$ is brought close enough to 0. And in odd numbered stages where $S_1$ is attached, the block entropy rate of $A_{2,1}^X$ relative to $A_{2,0}^X$ is brought close enough to 0. $S_0$ and $S_1$ are diluted sequences \cite{Dai2001} having finite state dimension equal to $1/2$. By carefully controlling the lengths of each individual stage and using the concavity of the Shannon entropy, we ensure that the block entropy rate of the whole infinite sequence $X$ has limit inferior equal to $1/2$.   Hence from all the conditions ensured during the stagewise construction, we get a sequence $X$ such that $\dim_{FS}(X) = 1/2$, $\dim_{FS}(A_{2,0}^X) = 0$ and $\dim_{FS}(A_{2,1}^X) = 0$, which proves the required result.
	
	The individual stage lengths are fixed such that they are long enough to satisfy a set of conditions. In order to precisely state these conditions, we require the following definitions.
		
	For $n\geq 0$, $\epsilon > 0$ and a finite string $\omega$, let $N_0^n(\omega, \epsilon)$ denote the smallest integer such that for all $m > N_0^n(\omega, \epsilon)$ and $\ell \leq n$, $H_\ell(A_{2,0}^{\omega S_0[:m]}) \leq \epsilon$.
	Similarly, let $N_1^n(\omega, \epsilon)$ denote the smallest integer such that for all $m > N_1^n(\omega, \epsilon)$ and $\ell \leq n$, $H_\ell(A_{2,1}^{\omega S_1[:m]} ) \leq \epsilon$. Note that this implies that, $H_\ell(A_{2,1}^{\omega S_1[:m]} \;\vert\; A_{2,0}^{\omega S_1[:m]} ) \leq \epsilon$ for every $\ell \leq n$ and $m > N_1^n(\omega, \epsilon)$.
	
	For $i \in \{0,1\}$ and $\omega \in \Sigma^*$, let $M_i^n(\omega,\epsilon)$ denote the smallest integer such that for all $m > M_i^n(\omega,\epsilon)$ and $l \leq n$, 
	\begin{equation}\label{eqn:something1}
		\lvert H_\ell(\omega S_i[:m]) - 1/2 \rvert \leq \epsilon.
	\end{equation}

	And for $i \in \{0,1\}$, let $J_i^n(\epsilon)$ denote the smallest integer such that for all $m \geq J_i^n(\epsilon)$ and $l \leq n$, $H_\ell(S_i[:m]) \geq 1/2 - \epsilon$.
	
	Using the above definitions, we now give the construction of $X$ in full detail. Initially let $\omega$ be the empty string. During stage $n \geq 0$, we set $w_n$ equal to $S_{n \Mod 2} [:m]$ where $m$ is the smallest number such that
	
	\begin{align*}
		m \geq \max \left\{N_{n \Mod 2}^n\left(\omega, \frac{1}{2^n}\right) , M_{n \Mod 2}^n\left(\omega, \frac{1}{2^{n+1}}\right), 2^{n+1} J_{(n+1) \Mod 2}^n\left(\frac{1}{2^{n+1}}\right)\right\}
	\end{align*}
	
	and $\lvert \omega \rvert +m$ is a multiple of every $\ell \leq n$. Now we append $w_n$ at the end and set $\omega$ be equal to $\omega w_n$. Similar steps are performed at every future stage $n$. The final sequence $X$ we obtain is the infinite string $w_0w_1w_2 \dots$. Now we show that $X$ satisfies all the required properties.
	
	Consider any $\ell \geq 1$, since by construction we have,
	$H_\ell(A_{2,0}^{w_0w_1 \dots w_n}) \leq \frac{1}{2^n}$, it follows that
	
	\begin{align*}
		\dim_{FS}(A_{2,0}^X)  &= \lim_{\ell \rightarrow \infty} \liminf_{m \rightarrow \infty} H_\ell(A_{2,0}^X[:m])\\
		&\leq \lim_{\ell \rightarrow \infty} \liminf_{n \rightarrow \infty} H_\ell (A_{2,0}^{w_0w_1 \dots w_n})\\
		&\leq \lim_{\ell \rightarrow \infty} \liminf_{n \rightarrow \infty} \frac{1}{2^n}\\
		&=0.
	\end{align*}
	
	A similar argument shows that $\dim_{FS}^{A_{2,0}^X}(A_{2,1}^X) = 0$.
	
	Now we show that $\dim_{FS}(X) = \frac{1}{2}$, which finishes the proof of the lemma. It is enough to show that for any $\ell$, 
	\begin{equation}\label{eqn:something4}
		\liminf_{n \rightarrow \infty} H_\ell(X[:n]) = \frac{1}{2}.
	\end{equation}
	
	By the construction of $X$, we have that $H_\ell(w_0w_1 \dots w_n) \leq \frac{1}{2} + 2^{-(n+1)}$ for every $\ell \leq n$. This implies that $\liminf_{n \rightarrow \infty} H_\ell(X[:n]) \leq \frac{1}{2} + 2^{-(n+1)}$. Letting $n \rightarrow 0$ we get 
	\begin{equation} \label{eqn:something3}
		\liminf_{n \rightarrow \infty} H_\ell(X[:n]) \leq \frac{1}{2}.
	\end{equation}
	
	For $\ell \geq 1$ and any positive real number $\epsilon$, let $n(\epsilon, \ell)$ be the smallest integer such that $2^{-n(\epsilon, \ell)} \leq \epsilon$ and $n(\epsilon, \ell) \geq \ell$. Define $k(\epsilon, \ell)$ to be equal to $\lvert w_0w_1 \dots w_{n(\epsilon, \ell)}\rvert$. Now we show that for any $m \geq k(\epsilon, \ell)$, $H_\ell(X[:m]) \geq \frac{1}{2} - \epsilon$. Let $n' \geq n(\epsilon, \ell)$ be such that $\lvert w_0w_1 \dots w_{n'}\rvert \leq m \leq \lvert w_0w_1 \dots w_{n'}w_{n'+1}\rvert $.
	
	If $m = \lvert w_0w_1 \dots w_{n'}\rvert$, then from the construction of $X$ and condition (\ref{eqn:something1}), it follows that $H_\ell (X[:m]) \geq \frac{1}{2} - \epsilon$. Now we consider the case when $m > \lvert w_0w_1 \dots w_{n'}\rvert$. Let $\alpha_m = w_{n'+1}[:m -  \lvert w_0w_1 \dots w_{n'}\rvert]$, that is $\alpha_m$ is the suffix of $X[:m]$ containing $w_{n'+1}$.
	
	For any $\ell \leq n$ and finite string $z \in \Sigma^*$, let $\mathbb{P}_{\ell}(\cdot, z)$ denote the probability distribution on $\Sigma^{\ell}$ such that for any $x \in \Sigma^{\ell}$, $\mathbb{P}_\ell (x,z)=P(x,z)$. Since $\lvert w_0w_1 \dots w_{n'}\rvert$ is ensured to be a multiple of every $\ell \leq n'$ during stage $n$, we have,
	
	\begin{align*}
		\mathbb{P}_{\ell}(x,X[:m]) = \frac{\lvert w_0w_1 \dots w_{n'}\rvert}{m}\mathbb{P}_{\ell} (x, w_0 w_1 \dots w_{n'}) + \frac{m-\lvert w_0w_1 \dots w_{n'}\rvert}{m}\mathbb{P}_{\ell} (x, \alpha_m)
	\end{align*}
	for every $\ell \leq n$ and $x \in \Sigma^{\ell}$. Therefore, using the concavity of the Shannon entropy, it follows that
	\begin{equation} \label{eq:something2}
		H_{\ell}(X[:m]) \geq \frac{\lvert w_0w_1 \dots w_{n'}\rvert}{m} H_\ell( w_0w_1 \dots w_{n'}) +  \frac{m - \lvert w_0w_1 \dots w_{n'}\rvert}{m}  H_\ell (\alpha_m).
	\end{equation}
	
	If $\lvert \alpha_m \rvert \leq J_{(n'+1) \mod 2}^{n'} \; (\frac{1}{2^{n'+1}})$, we have $\frac{\lvert w_0w_1 \dots w_{n'}\rvert}{m} \geq 1- \frac{1}{2^{n'+1}}$. Substituting the above in equation (\ref{eq:something2}), it follows that,
%

	\begin{align*}
		H_{\ell}(X[:m]) &\geq  \left(1- \frac{1}{2^{n'+1}} \right)  H_\ell( w_0w_1 \dots w_{n'}) \\
		&\geq \left(1- \frac{1}{2^{n'+1}} \right) \left(\frac{1}{2} - \epsilon\right) \\
		&=\frac{1}{2}-\epsilon-\epsilon \cdot \frac{1}{2^{n'+1}} - \frac{1}{2^{n'+2}}\\
		&\geq \frac{1}{2} - 2 \epsilon -\frac{1}{2^{n'+2}} \\
		&\geq  \frac{1}{2} - 3\epsilon.
	\end{align*}
		
	Finally, we consider the case when $\lvert \alpha_m \rvert > J_{(n'+1) \mod 2}^{n'}$, From the definition of $ J_{(n'+1) \mod 2}^{n'}$, we get that $H_\ell(\alpha_m) \geq \frac{1}{2} - \frac{1}{2^{n'+1}} $. 
	
	Using the above inequality in (\ref{eqn:something1}), it follows that,
	
	\begin{align*}
		H_{\ell}(X[:m]) &\geq  \frac{\lvert w_0w_1 \dots w_{n'}\rvert}{m} \Big(\frac{1}{2} - \frac{1}{2^{n'+1}}\Big) + \frac{m - \lvert w_0w_1 \dots w_{n'}\rvert}{m}  \Big(\frac{1}{2} - \frac{1}{2^{n'+1}}\Big) \\
		&= \frac{1}{2} - \frac{1}{2^{n'+1}} \\
		&\geq  \frac{1}{2} - \epsilon. \\
		&> \frac{1}{2} - 3\epsilon.
	\end{align*}
	
	Hence we have shown that for any $m \geq k(\epsilon,\ell)$,
	$H_{\ell}(X[:m]) \geq  \frac{1}{2} - 3\epsilon$, this implies that,
	
	\begin{align*}
		\liminf_{n \rightarrow \infty} H_{\ell}(X[:n]) \geq \frac{1}{2} - \epsilon.
	\end{align*}
	
	Letting $\epsilon \rightarrow 0 $ in the above inequality and using (\ref{eqn:something3}), we obtain (\ref{eqn:something4}), which completes the proof of the lemma.
\end{proof}

We now give an upper bound for $\dim_{FS}(X)$ using Theorem \ref{thm:fsstrongrelativedimensionentropyrate}. We require the following upper bound on the limit inferior of the sum of finitely many real sequences which easily follows using routine real analytic arguments.
\begin{lemma}
\label{lem:liminfupperbound}
Let $\langle a_{0,n} \rangle_{n =0}^{\infty},\langle a_{1,n} \rangle_{n =0}^{\infty},\dots \langle a_{k-1,n} \rangle_{n =0}^{\infty}$ be any finite collection of sequences of real numbers. Then,
\begin{align*}
\liminf\limits_{n \to \infty} \left( \sum\limits_{i=0}^{k-1} a_{i,n} \right)\leq \min\limits_{0 \leq m \leq d-1}\left\{ \liminf\limits_{n \to \infty}	a_{m,n}+ \sum\limits_{i \neq m} \limsup\limits_{n \to \infty}	a_{i,n}  \right\}.
\end{align*}
\end{lemma}

Now Theorem \ref{thm:fsrelativedimensionentropyrate}, Lemma \ref{lem:fsdofproductandfsdofinterleaving} and Lemma \ref{lem:liminfupperbound} yields the following upper bound for the finite state dimension of $X$. 

 \begin{lemma}
\label{lem:mainupperboundinequality}
	For any $X \in \Sigma^\infty$, $d \in \N$ and $\sigma \in \mathrm{Bij}(d)$, then $\dim_{FS}(X)$ is less than or equal to
	\begin{align*}
	 \min\limits_{0 \leq m \leq d-1} \frac{1}{d} \left(\dim_{FS}^{A_{d,\sigma(0)},\dots A_{d,\sigma(m-1)}}(A_{d,\sigma(m)}) +\sum\limits_{i \neq m} \Dim_{FS}^{A_{d,\sigma(0)},\dots A_{d,\sigma(i-1)}}(A_{d,\sigma(i)}) \right)
	\end{align*}
\end{lemma}

There is a minor technical point to be noted in the above statement. When $i$ (or $m$) is equal to $0$, $A_{d,\sigma(0)},A_{d,\sigma(1)},\dots A_{d,\sigma(i-1)}=A_{d,\sigma(0)},A_{d,\sigma(1)},\dots A_{d,\sigma(-1)}$ is not defined, because $\sigma(-1)$ is not defined. So, $\dim_{FS}^{A_{d,\sigma(0)},A_{d,\sigma(1)},\dots A_{d,\sigma(0-1)}}(A_{d,\sigma(0)})$ is a placeholder for $\dim_{FS}(A_{d,\sigma(0)})$, which is used in order to make the statement of Lemma \ref{lem:mainupperboundinequality} notationally simpler. A similar definition applies in the case of $\Dim_{FS}^{A_{d,\sigma(0)},A_{d,\sigma(1)},\dots A_{d,\sigma(0-1)}}(A_{d,\sigma(0)})$.
\begin{proof}[Proof of Lemma \ref{lem:mainupperboundinequality}]
	The proof is identical to that of Lemma \ref{lem:mainlowerboundinequality} till the derivation of equation \ref{eqn:mainlowerboundeqn2} using the chain rule of Shannon entropy. In the next step we get,
	\begin{align*}
		&\dim_{FS}(X) = \dim_{FS}(A_{d,\sigma(0)} \times A_{d,\sigma(1)} \times \dots A_{d,\sigma(d-1)})\\
		&= \frac{1}{d} \lim\limits_{l \to \infty} \liminf\limits_{k \to \infty}  \sum\limits_{i=0}^{d-1}  H_l(A_{d,\sigma(i)}[:kl] \mid A_{d,\sigma(0)}[:kl], A_{d,\sigma(1)}[:kl] \dots A_{d,\sigma(i-1)}[:kl] ).
	\end{align*}
	In the above, when $i=0$, we use the quantity $H_l(A_{d,\sigma(i)}[:kl] \mid A_{d,\sigma(0)}[:kl], A_{d,\sigma(1)}[:kl] \dots A_{d,\sigma(-1)}[:kl] )$ as a placeholder for $H_l(A_{d,\sigma(0)}[:kl] )$ for notational simplicity, similar to our remark following the statement of the lemma. Fix an arbitrary $m \in \{0,1,2,\dots d-1\}$. Now, using Theorem \ref{thm:multifsrelativedimensionentropyrate}, the upper bound from Lemma \ref{lem:liminfupperbound} and by distributing $\lim\limits_{l \to \infty}$ over the inner terms, we get,
	
	\begin{align*}
		&\dim_{FS}(X) = \dim_{FS}(A_{d,\sigma(0)} \times A_{d,\sigma(1)} \times \dots A_{d,\sigma(d-1)})\\
		&\leq \frac{1}{d} \lim\limits_{l \to \infty}  \liminf\limits_{k \to \infty}H_l(A_{d,\sigma(m)}[:kl] \mid A_{d,\sigma(0)}[:kl], A_{d,\sigma(1)}[:kl] \dots A_{d,\sigma(m-1)}[:kl] ) \\
		&+ \frac{1}{d}\sum\limits_{i \neq m}  \lim\limits_{l \to \infty}\limsup\limits_{k \to \infty}H_l(A_{d,\sigma(i)}[:kl] \mid A_{d,\sigma(0)}[:kl], A_{d,\sigma(1)}[:kl] \dots A_{d,\sigma(i-1)}[:kl] ).
	\end{align*}
	The required inequality now follows by observing that,
	\begin{align*}
	&\dim_{FS}^{A_{d,\sigma(0)},A_{d,\sigma(1)},\dots A_{d,\sigma(m-1)}}(A_{d,m})\\
	 &= \lim\limits_{l \to \infty}  \liminf\limits_{k \to \infty}H_l(A_{d,\sigma(m)}[:kl] \mid A_{d,\sigma(0)}[:kl], \dots A_{d,\sigma(m-1)}[:kl] )
	\end{align*}
	and,
	\begin{align*}
	&\Dim_{FS}^{A_{d,\sigma(0)},A_{d,\sigma(1)},\dots A_{d,\sigma(i-1)}}(A_{d,\sigma(i)}) \\
	&= \lim\limits_{l \to \infty}  \limsup\limits_{k \to \infty}H_l(A_{d,\sigma(i)}[:kl] \mid A_{d,\sigma(0)}[:kl], \dots A_{d,\sigma(i-1)}[:kl] ).
	\end{align*}
\end{proof}

As in the case of Lemma \ref{lem:mainlowerboundinequality}, it follows from Theorem \ref{thm:strongerwallstheoremwithconverse} that the upper bound in Lemma \ref{lem:mainupperboundinequality} is tight when  $X$ is normal. To conclude we show that the inequality in Lemma \ref{lem:mainupperboundinequality} cannot be replaced in general with an equality. 

\begin{lemma}
\label{lem:lowerboundistight}
	There exists an $X \in \Sigma^\infty$ such that,
	\begin{align*}
	 \dim_{FS}(X)	 < \frac{1}{2}\left(\dim_{FS}(A_{2,0})+\Dim_{FS}^{A_{2,0}}(A_{2,1}) \right)
	\end{align*}
\end{lemma}
\begin{proof}[Proof sketch of Lemma \ref{lem:lowerboundistight}]
	The construction of $X$ is similar to that in the proof of Lemma \ref{lem:upperboundistight}. Let $Y$ be any fixed normal sequence. $X$ is constructed in a stage-wise manner by appending \emph{long enough} prefixes of $Y$ and $0^\infty$ in an alternating stage-wise manner so that the following properties are satisfied:
	\begin{enumerate}
		\item The block entropy rate of $X$ \emph{oscillates} between $0$ and $1$ so that $\dim_{FS}(X)=0$.
		\item The block entropy rates of $A_{2,0}$ and $A_{2,1}$ gets arbitrarily close to $0$ in subsequent stages so that $\dim_{FS}(A_{2,0})=0$ and $\dim_{FS}^{A_{2,0}}(A_{2,1})=0$.
		\item The block entropy rate of $A_{2,1}$ relative to $A_{2,0}$ gets arbitrarily close to $1$ in subsequent stages so that $\Dim_{FS}^{A_{2,0}}(A_{2,1})=1$.
	\end{enumerate}
	With the above conditions satisfied, we have $\dim_{FS}(X)=0$ and $ (\dim_{FS}(A_{2,0})+\Dim_{FS}^{A_{2,0}}(A_{2,1}) )/2=1/2$. Hence, the constructed sequence $X$ satisfies the required property.
\end{proof}

\section{A stronger Wall's theorem on A. P. subsequences with a perfect converse}
\label{sec:strongerwallstheoremwithconverse}
D.~D.~Wall in \cite{Wall1950} proved that a number
$x=0.x[0]x[1]x[2]x[3]\dots$ is normal if and only if for every $d \geq
1$ and $a \geq 0$, $0.x[a]x[a+2d]x[a+3d]x[a+4d]\dots$ is a normal
number. We state the equivalent theorem in $\Sigma^\infty$ below.
\begin{theorem}[\cite{Wall1950}]
\label{thm:originalwallstheorem}
If $X \in \Sigma^\infty$ is a normal sequence then for every $d \geq
1$ and $a \geq 0$, $X[a]\;X[a+d]\;X[a+2d]\;X[a+3d]\;X[a+4d]\dots$ is a
normal sequence.
\end{theorem}
We remark that it is enough to consider $a \in \{0,1,2,\dots d-1\}$
since replacing $a$ with $a \Mod d$ only prepends finitely many
characters to the subsequence (and therefore has no effect on the
normality or the finite-state dimension of the subsequence). If $d=1$,
then the converse direction is trivial since the A. P. subsequence with
$a=0$ and $d=1$ is the sequence $X$ itself. Therefore, the converse
direction is interesting only when $d$ takes values strictly greater
than $1$. Hence, the new converse question is the following: \emph{If
for every $d \geq 2$ and $a \in \{0,1,2,\dots d-1\}$,
$X[a]\;X[a+d]\;X[a+2d]\;X[a+3d]\;X[a+4d]\dots$ is a normal sequence then
is $X$ a normal sequence?}.

The answer is negative. Vandehey (\cite{Vandehey2019}) gives the
following counterexample: fix any normal sequence
$X=X[0]X[1]X[2]X[3]\dots$ and consider the \emph{doubled sequence} defined as follows:
$X[0]X[0]X[1]X[1]X[2]X[2]X[3]X[3]\dots$. It is straightforward to
verify that every A. P. subsequence of the \emph{doubled sequence}
with $d \geq 2$ is normal, but the \emph{doubled sequence} is itself
non-normal.

Weiss \cite{weiss1971normal}, Kamae \cite{Kamae1973} and Vandehey
\cite{Vandehey2019} show that expanding the set of subsequences along
which normality is investigated can yield interesting answers in the
converse direction. Weiss and Kamae (\cite{weiss1971normal} and
\cite{Kamae1973}) showed that a number is normal if and only normality
is preserved along every \emph{deterministic} subsequence with
positive asymptotic lower density. In a recent work, Vandehey
\cite{Vandehey2019} proved a nearly sharp converse to the Wall's
theorem. Vandehey considers collections of subsequences such that for
any $\epsilon$, the collection contains a subsequence with asymptotic
lower density greater than $1-\epsilon$. Theorem 1.3 in
\cite{Vandehey2019} shows that preservation of normality along all
subsequences in such a collection implies the normality of the
original number. Theorem 1.4 from the same paper shows that this
converse to the Wall's theorem is close to being sharp.

 In this section we show that the inequalities established in section
 \ref{sec:fsrelativedimensionandapsubsequences} yields a
 \emph{stronger} forward direction for Wall's theorem (Theorem
 \ref{thm:originalwallstheorem}). We show that if sequence $X$ is
 normal, then for any $d \geq 2$ and $a \in \{0,1,2,\dots d-1\}$, the
 A. P. subsequences $X[a]\;X[a+d]\;X[a+2d]\;X[a+3d]\;X[a+4d]\dots$ is
 \emph{relatively normal} with respect to every other A. P. subsequence
 with the same common difference $d$. The results in section
 \ref{sec:fsrelativedimensionandapsubsequences} also yields a perfect
 converse to the \emph{stronger} forward direction. i.e, we show that
 if every A. P. subsequence of $X$ with $d \geq 2$ is \emph{relatively
 normal} with respect to every other A. P. subsequence with the same
 common difference $d$, then $X$ is a normal sequence.
 
 For proving the forward direction we require the following lemma.
  
\begin{lemma}
\label{lem:apsubsequencedimensionlowerbound}
	For any $X \in \Sigma^\infty$, $d \in \N$ and $j \in \{0,1,2,\dots d-1\}$,
	\begin{align*}
		\dim_{FS}(A_{d,j}) \geq \dim_{FS}^{\{A_{d,k} \mid k \neq j\}}(A_{d,j}) \geq d\left(\dim_{FS}(X)-\frac{d-1}{d}\right).
	\end{align*}
\end{lemma}
\begin{proof}
	Choose any $\sigma \in \mathrm{Bij}(d)$ such that
%
%
	$\sigma(m)=j$. Now, from Lemma \ref{lem:mainupperboundinequality} we get that,
%
	\begin{align*}
	\dim_{FS}^{\{A_{d,k} \mid k \neq j\}}(A_{d,j}) \geq d \left( \dim_{FS}(X)-\frac{1}{d} \sum\limits_{i \neq m} \Dim_{FS}^{A_{d,\sigma(0)},\dots A_{d,\sigma(i-1)}}(A_{d,\sigma(i)})\right).
	\end{align*}
	The required result now follows using the observation that,
	\begin{align*}
	 \Dim_{FS}^{A_{d,\sigma(0)},\dots A_{d,\sigma(i-1)}}(A_{d,\sigma(i)}) \leq 1.
	\end{align*}
	for every $i \neq m$.
\end{proof}

Now, we give an immediate corollary of Lemma \ref{lem:apsubsequencedimensionlowerbound} which is useful in the later sections.
\begin{corollary}
\label{cor:apsubsequencedimensionboundcorollary}
Let $X \in \Sigma^\infty$, $d \in \N$ and $j \in \{0,1,2,\dots d-1\}$. If for some $\epsilon>0$,
\begin{align*}
	\dim_{FS}(X) \geq \frac{d-1}{d} + \epsilon,
\end{align*}
then for every $j$, $\dim_{FS}(A_{d,j}) \geq \dim_{FS}^{\{A_{d,k} \mid k \neq j\}}(A_{d,j}) \geq d\epsilon$.
\end{corollary}

The following is the restatement of Theorem
\ref{thm:originalwallstheorem} in terms of finite-state dimension
which easily follows from the entropy characterization of finite-state
dimension from \cite{bourke2005entropy}.
\begin{theorem*}[Restatement of Theorem \ref{thm:originalwallstheorem}]
A sequence $X \in \Sigma^\infty$ is such that $\dim_{FS}(X)=1$ if and
only if for every $a \in \{0,1,2,\dots d-1\}$ and $d \geq 2$,
$\dim_{FS}(A_{d,a})=1$.
\end{theorem*}

As a consequence of Corollary \ref{cor:apsubsequencedimensionboundcorollary}, by setting $\epsilon=1/d$, we obtain the \emph{stronger} forward direction of
Theorem \ref{thm:originalwallstheorem}.
\begin{lemma}
\label{lem:wallstheoremforwarddirection}
Let $X \in \Sigma^\infty$ be any normal sequence. Then, for every $d \geq 2$ and $a \in \{0,1,2,\dots d-1\}$, $\dim_{FS}(A_{d,a}) = \dim_{FS}^{\{A_{d,k} \mid k \neq a \}}(A_{d,a}) =1$. 
\end{lemma}
The above statement strengthens the forward direction of Theorem
\ref{thm:originalwallstheorem} because it claims that \emph{if a
sequence is normal then each of its A. P. subsequences are normal and are
also relatively normal with respect to the other A. P. subsequences
having the same common difference}.

The conclusion in Lemma \ref{lem:wallstheoremforwarddirection} is
strong enough that using our lower bound inequality (Lemma
\ref{lem:mainlowerboundinequality}), we get a \emph{perfect} converse
to Lemma \ref{lem:wallstheoremforwarddirection}.
\begin{lemma}
\label{lem:wallstheoremconverse}
Let $X \in \Sigma^\infty$. If for every $d \geq 2$ and $a \in
\{0,1,2,\dots d-1\}$, $\dim_{FS}^{\{A_{d,k} \mid k \neq a \}}(A_{d,a})
=1$, then $\dim_{FS}(X)=1$.
\end{lemma}
Lemma \ref{lem:wallstheoremconverse} follows using $\dim_{FS}^{A_{d,0},A_{d,1},\dots A_{d,a-1}}(A_{d,a}) \geq
\dim_{FS}^{\{A_{d,k} \mid k \neq a \}}(A_{d,a})$ and Corollary
\ref{cor:maininequalitycorollary1} when $r=1$. 

Combining Lemma \ref{lem:wallstheoremforwarddirection} and
\ref{lem:wallstheoremconverse}, we get the following \emph{stronger}
Wall's theorem with a \emph{perfect} converse.
\begin{theorem}
\label{thm:strongerwallstheoremwithconverse}	
A sequence $X \in \Sigma^\infty$ is normal (equivalently
$\dim_{FS}(X)=1$) if and only if for every $d \geq 2$ and $a \in
\{0,1,2,\dots d-1\}$, $\dim_{FS}^{\{A_{d,k} \mid k \neq a \}}(A_{d,a})
=1$.
\end{theorem}
In other words, we have shown that a sequence $X \in \Sigma^\infty$ is
normal if and only if for every $d \geq 2$, the A. P. subsequences of $X$
with common difference $d$ are normal and are also relatively normal
with respect to the other A. P. subsequences with the same common
difference $d$.

\section{van Lambalgen's Theorem for Finite-State Dimension}
\label{sec:fsvanlambalgentheorems}
van Lambalgen, in his thesis, \cite{lambalgen1987random} showed that
relative Martin-L\"of randomness is symmetric. 
\begin{theorem}[\cite{lambalgen1987random}] 
Let $A,B \in \Sigma^\infty$. $A$ is Martin-L\"of random and $B$ is
Martin-L\"of random relative to $A$ if and only if $A \oplus B$ is
Martin-L\"of random.
\end{theorem}
This symmetry fails in other randomness settings like Schnorr
randomness, computable randomness and resource-bounded randomness
\cite{Yu2007}, \cite{Bauwens2020}, \cite{DowneyHirschfeldt2010},
\cite{Chakraborty2017}. In this section, we show that relative
normality is symmetric, thus establishing an analogue of van
Lambalgen's theorem for normality. But, for finite-state dimensions
less than 1, both directions of the theorem fails to hold in general.
A sequence $X \in \Sigma^\infty$ is a \emph{regular} sequence if
$\dim_{FS}(X)=\Dim_{FS}(X)$ (see, for example, \cite{Miller2011},
where the notion is defined for effective dimension). We show that for
the class of regular sequences, the forward direction of van
Lambalgen's Theorem is true.

Utilizing the results we established in Sections
\ref{sec:fsrelativedimensionandapsubsequences} and
\ref{sec:strongerwallstheoremwithconverse}, we first show that an
analogue of van Lambalgen's theorem holds for normality, \emph{i.e.}
the case when finite-state dimensions are 1.

\begin{theorem}
  \label{thm:vanLambalgen}
Let $A$, $B$ $\in \Sigma^\infty$. $A$ is normal and $B$ is normal
relative to $A$ if and only if $A \oplus B$ is normal.
\end{theorem}
\begin{proof}
$A$ and $B$ are two A.P. subsequences of $A \oplus B$ with common
  difference 1. By Corollary \ref{cor:maininequalitycorollary1} it
  follows that $\dim_{FS}(A \oplus B)=1$, \emph{i.e.}, $A \oplus B$ is
  normal.

Conversely, suppose that $A \oplus B$ is normal, \emph{i.e.}
$\dim_{FS}(A \oplus B)=1$. By Lemma
\ref{lem:wallstheoremforwarddirection}, $\dim^A_{FS}(B)=1$ and
$\dim_{FS}^B(A)=1$. Since $\dim_{FS}(A) \ge \dim^B_{FS}(A)$ by Lemma
\ref{lem:apsubsequencedimensionlowerbound}, it follows that
$\dim_{FS}(A)=1$.
\end{proof}

Thus relative normality is symmetric. We may conjecture that this
generalizes in the following ideal form.

{\bf Ideal Claim.} Let $A$, $B$ $\in \Sigma^\infty$. Then for any $r
\in [0,1]$, $\dim_{FS}(A)=r$ and $\dim^A_{FS}(B)=r$ if and only if
$\dim_{FS}(A \oplus B)=r$.

Both the forward and converse directions in the above claim are false
for general sequences.  However, we conclude by showing that if $A$
and $B$ are ``regular sequences'', then the forward direction of the
ideal claim holds. This is analogous to the failure of the converse
direction of van Lambalgen's theorem in certain notions of randomness.

Setting $A$ and $B$ to be the two A. P. subsequences (with $d=2$) of $X$
from Lemma \ref{lem:upperboundistight}, it readily follows that the
forward direction in the above claim is false for general $A$ and
$B$. Now let $A$ be $0^\infty$ and $B$ be any normal number. Now, $A
\oplus B$ is the diluted sequence (\cite{Dai2001}) with dimension
equal to $1/2$. But, $\dim_{FS}(A)=0$ and it follows from Lemma \ref{lem:basicproperties} that
$\dim_{FS}^{A}(B)=1$. Therefore the converse direction of the claim is
also false for general sequences. It is easy to verify from the
construction in the proof of Lemma \ref{lem:upperboundistight} that
sequences $A$ and $B$ given in the counterexample for the forward
direction are both non-regular sequences. This leads to the question
whether the forward direction in the ideal claim is true for $A$ and
$B$ that are regular sequences. We answer this question in the
affirmative as a consequence of the following lemma.
\begin{lemma}
\label{lem:regularsequencesdimensionequation}
	Let $A,B \in \Sigma^\infty$. If $A$ is a regular sequence, then,
	\begin{align*}
		\dim_{FS}(A \oplus B) = \frac{1}{2} \left(\dim_{FS}(A)+\dim_{FS}^{A}(B) \right).
	\end{align*}
\end{lemma}
\begin{proof}
Since, $A$ and $B$ are the two A. P. subsequences (with $d=2$) of $A
\oplus B$, by choosing $\sigma$ to be the identity mapping, it follows
from Lemma \ref{lem:mainlowerboundinequality} that,
\begin{align*}
	\dim_{FS}(A \oplus B) \geq \frac{1}{2} \left(\dim_{FS}(A)+\dim_{FS}^{A}(B) \right).
\end{align*}
By setting $X= A \oplus B$ and choosing the identity mapping $\sigma$ from $\mathrm{Bij}(2)$, we obtain the following using Lemma \ref{lem:mainupperboundinequality},
\begin{align*}
\dim_{FS}(x)=\dim(A \oplus B) &\leq \frac{1}{2} \left(\Dim_{FS}(A)+\dim_{FS}^{A}(B) \right).
\end{align*}
The required inequality now follows by observing that $\Dim_{FS}(A)=\dim_{FS}(A)$ since $A$ is regular.
\end{proof}

Now we prove the forward direction of the ideal claim for regular sequences.
\begin{lemma}
	Let $A,B \in \Sigma^\infty$ such that $A$ is a regular sequence and let $r \in [0,1]$. If $\dim_{FS}(A)=r$ and $\dim_{FS}^{A}(B)=r$ then, $\dim_{FS}(A \oplus B)=r$.
\end{lemma}

The converse of the ideal claim is however false even for $A$ and $B$
that are both regular sequences. Let $A$ be $0^\infty$ and $B$ be any normal number. Both of these sequences are regular sequences. However, as we noted above, $\dim_{FS}(A \oplus B)=1/2$, but $\dim_{FS}(A)=0$ and $\dim_{FS}^A(B)=1$.


\section{Relation to Mutual Dimension}
Case and Lutz \cite{case2021finite} introduce finite-state mutual
dimension as a finite state analogue of mutual dimension
\cite{DowneyHirschfeldt2010}. Intuitively, it represents the density
of finite-state information shared between two sequences.

 For two sequences $X,Y \in \Sigma^\infty$, the \emph{finite-state
 mutual dimension} between X and Y,
\begin{align*}
	\mathrm{mdim}_{FS}(X;Y) = 
	\lim_{\ell \rightarrow \infty} \liminf\limits_{k\rightarrow \infty} I_\ell(X[:k\ell]\;;\;Y[:k\ell])
\end{align*}
where for two finite sequences $X,Y \in \Sigma^{k\ell}$,
\begin{equation}
	I_\ell(X\;;\;Y) = H_\ell(X) - H_\ell(X\;|\;Y).	\label{eq:n1}
\end{equation}

 The \emph{finite-state strong mutual dimension} $\mathrm{Mdim}_{FS}(X;Y)$ between $X$ and $Y$ is defined similarly as $\mathrm{mdim}_{FS}(X;Y)$ by replacing the $\liminf$ in the definition with a $\limsup$. Using \eqref{eq:n1} and routine real analytic arguments it follows that the following relationships between finite-state relative dimension and finite-state mutual dimension hold:

\begin{enumerate}
	\item $\dim_{FS}(X) - \Dim_{FS}^Y(X) \leq \mathrm{mdim}_{FS}(X:Y) \leq \dim_{FS}(X) - \dim_{FS}^Y(X)$  
	\item $\Dim_{FS}(X) - \Dim_{FS}^Y(X) \leq \mathrm{Mdim}_{FS}(X:Y) \leq \Dim_{FS}(X) - \dim_{FS}^Y(X)$  
\end{enumerate}

If the sequences $X$ and $Y$ are regular, then by definition, $\dim_{FS}(X) = \Dim_{FS}(X)$ and $\dim_{FS}^Y(X) = \Dim_{FS}^Y(X)$ and so,

\begin{align*}
\dim_{FS}^Y(X) =  \dim_{FS}(X) -  \mathrm{mdim}_{FS}(X:Y) 
\end{align*}

\section{Relation to finite-state independence}
  Becher, Carton and Heiber \cite{Becher2018} define finite-state
  independence between two strings. They use a joint compression model
  composing of a one-to-one finite-state transducer with auxiliary
  input.  Two strings $X$ and $Y$ are said to be finite-state
  independent when one does not help to compress the other in this
  model.  However, due to the presence of $\varepsilon$-transitions,
  the input tape of $X$ and $Y$ need not be read in tandem. Our model
  is more restrictive, hence more pairs of sequences are relatively
  finite-state random in our model. This is a consequence of the
  theorems in the previous section.
  
  \begin{lemma}
There are sequences $X$ and $Y$ that are not finite-state independent
but $X$ and $Y$ are relatively normal, and $X \oplus Y$ is normal.
  \end{lemma}
  \begin{proof}
By Theorem 4.3 in \cite{Becher2018}, there exists two normal strings
$X$ and $Y$ such that $X \oplus Y$ is normal but $X$ and $Y$ are not
finite-state independent. Since $\dim_{FS}(X \oplus Y) = 1$, by Theorem
\ref{thm:vanLambalgen}, we have $\dim^X(Y) = 1$.
  \end{proof}

The converse question, \emph{i.e.} whether finite-state independence
implies relative randomness, remains open.

\bibliographystyle{plain} 
\bibliography{main}

\end{document}

%% file: commands.tex
\usepackage{amsfonts}
\usepackage{amsmath} 
\usepackage{amssymb}
\usepackage{enumerate}

\newcommand{\PROOF}{\begin{proof}}

\newcommand{\QED}{\end{proof}}

\newcommand{\N}{\mathbb{N}}

\renewcommand{\dim}{{\mathrm{dim}}}

\newcommand{\Dim}{{\mathrm{Dim}}}

\usepackage[svgnames]{xcolor}

\def\<{\left\langle}
\def\>{\right\rangle}